\documentclass[a4paper,12pt]{article}
\usepackage{amsmath,amssymb,amsfonts,amsthm}
\newtheorem{theorem}{Theorem}[section]
\newtheorem{lemma}[theorem]{Lemma}

\usepackage{graphicx}

\newcommand{\Rt}{\mathbb{R}^3}

\newcommand{\den}{\varepsilon}
\newcommand{\es}{\varepsilon_\sigma}

\newcommand{\eo}{\varepsilon_\omega}
\newcommand{\ms}{m_\sigma}

\newcommand{\mo}{m_\omega}

\newcommand{\Rdm}{\mathbb{R}^2_{+}}

\newcommand{\Ld}{\Delta}

\newcommand{\Ldt}{{^{(3)}\Delta}}

\newcommand{\ck}{\mathcal{L}}

\newcommand{\po}{d}
\newcommand{\ps}{p}

\newcommand{\Dk}{\Xi}

\title{On the linear stability of the extreme Kerr black hole under axially
  symmetric perturbations}
\author{Sergio Dain and Ivan Gentile de Austria\\
 Facultad de Matem\'atica, Astronom\'{i}a y F\'{i}sica, FaMAF,\\
  Universidad Nacional de C\'ordoba,\\
  Instituto de F\'{\i}sica Enrique Gaviola, IFEG, CONICET,\\
  Ciudad Universitaria, (5000) C\'ordoba, Argentina.}

\begin{document}
\maketitle

\begin{abstract}
  We prove that for axially symmetric linear gravitational perturbations of the
  extreme Kerr black hole there exists a positive definite and conserved
  energy.  This provides a basic criteria for linear stability in axial
  symmetry. In the particular case of Minkowski, using this energy we also
  prove pointwise boundedness of the perturbation in a remarkable simple way.
\end{abstract}

\section{Introduction}
\label{s:introduction}
 
Recently there has been considerable progress on the long standing and central
open problem of black hole stability in General Relativity (see the review
articles \cite{Dafermos:2010hd} \cite{Dafermos:2008en} and reference therein).
The following three aspects of this problem motivated the present work. 

\textbf{(i) Non-modal stability of linear gravitational perturbations:} the
non-modal stability of linear gravitational perturbations for the Kerr black
hole still remains unsolved.  The works of Regge, Wheeler \cite{Regge:1957td},
Zerilli \cite{Zerilli:1970se} \cite{Zerilli:1974ai} and Moncrief
\cite{Moncrief:1975sb} determined the modal linear stability of gravitational
perturbations for the Schwarzschild black hole by ruling out exponential growth
in time for every individual mode.  The modal stability for the Kerr black hole
was proved by Whiting \cite{Whiting89} using the Teukolsky equation. However,
modal stability is not enough to exclude that general linear perturbations grow
unbounded in time (see, for example, the discussion in \cite{wald:1056} and
\cite{Dafermos:2010hd}). The study of black hole non-modal stability was
initiated by Kay and Wald in \cite{wald:1056} \cite{Kay:1987ax}. They prove
that solutions of the linear wave equation on a Schwarzschild black hole
background remain bounded by a constant for all time. An important ingredient
in this proof is the use of conserved energies to control the norm of the
solution. The analog of the Kay-Wald theorem on a large class of backgrounds
which includes the slow rotating Kerr black hole was first proved by Dafermos
and Rodniaski \cite{Dafermos:2008ys} and then, independently, in the special
case of slow rotating Kerr by Andersson and Blue \cite{andersson09}. In
\cite{Dafermos:2010hd} Dafermos and Rodniaski provide the essential elements of
the proof of this theorem for the general subextremal Kerr black
hole. Recently, this problem was finally solved in \cite{Dafermos:2014cua}.
For a complete list of references with important related works on this subject
see the review articles \cite{Dafermos:2010hd} \cite{Dafermos:2008en}
\cite{Finster:2008bg}. All these results concern the wave equation. For
gravitational perturbations the only non-modal stability result was given very
recently by Dotti \cite{Dotti:2013uxa} for the Schwarzschild black hole. There
are, so far, no results regarding the non-modal stability of the Kerr black
hole under linear gravitational perturbations.

\textbf{(ii) Stability and instability of extreme black holes:} extreme black
holes are relevant because they lie on the boundary between black holes and
naked singularities and hence it is expected that their study shed light on the
cosmic censorship conjecture.  Recently, Aretakis discovered certain
instabilities for extreme black holes \cite{Aretakis:2011ha}
\cite{Aretakis:2011hc}. These instabilities concern transverse derivatives of
the field at the horizon: a conservation law ensures that the first transverse
derivative of the field on the event horizon generically does not decay, this
implies that the second transverse derivative of the field generically grows
with time on the horizon. These instabilities were discovered first for the
scalar wave equation on the extreme Reissner-Nordstr\"om black hole, a similar
result also holds for the extreme Kerr black hole \cite{Aretakis:2012ei}
\cite{Aretakis:2011gz}.  These works were extended in several directions: for
generic extreme black holes and linear gravitational perturbations
\cite{Lucietti:2012sf}, for certain higher dimensional extreme vacuum black
holes \cite{Murata:2012ct}; for massive scalar field and for coupled linearized
gravitational and electromagnetic perturbations \cite{Lucietti:2012xr}, for a
test scalar field with a nonlinear self-interaction in the extreme Kerr
geometry \cite{Aretakis:2013dpa}. An interesting relation between these
instabilities and the Newman-Penrose constants was pointed out by Bizon and
Friedrich \cite{Bizon:2012we}. This relation was also independently observed by
Lucietti, Murata, Reall and Tanahashi \cite{Lucietti:2012xr}. Finally, a
numerical study of nonlinear evolution of this instability for spherically
symmetric perturbations of an extreme Reissner-Nordstr\"om black hole was
performed by Murata, H. S. Reall, and N. Tanahashi in \cite{Murata:2013daa}.

An important question regarding the dynamical behaviour of extremal black holes
is whether a non-extremal black hole can evolve to an extremal one at late
times. In \cite{Reiris:2013efa} Reiris proved that there exists arbitrary small
perturbations of the extreme black hole initial data that can not decay in time
into any extreme black hole. On the other hand, in \cite{Murata:2013daa} fine
tuned initial data are numerically constructed which settle to an extreme
Reissner-Nordstr\"om black hole. There is no contradiction between these two
results since they apply to different kind of data.  It is interesting to note
that the construction in \cite{Reiris:2013efa} relies on geometrical
inequalities between area and charges on trapped surfaces (see \cite{dain12}
and reference therein), in contrast in the spacetime considered in
\cite{Murata:2013daa} there are no trapped surfaces.

The discussion above concern instability of extreme black
holes. However, there are also stability results for this class of black
holes. The most relevant of them is that the solutions of the wave equations
remain pointwise bounded in the black hole exterior region
\cite{Aretakis:2011ha} (see also \cite{Dain:2012qw}).

\textbf{(iii) Non-linear stability:} the problem of the black hole non-linear
stability remains largely open (see the discussion in \cite{Dafermos:2010hd}
and reference therein).  The linear studies previously discussed are expected
to provide insight into the non-linear problem. However, this will be possible
only if they rely on techniques that can be suitable extended to the non-linear
regime. One of the most important of these techniques are the energy estimates.

\vspace{1cm}

The main result of this article is the following:
\begin{quote}
  \emph{For axially symmetric linear gravitational perturbation of the extreme
    Kerr black hole there exists an energy which is positive definite and
    conserved.}
\end{quote}

A precise version of this statement is given in  Theorem \ref{theorem}. In the
following we discuss the relation of this result with  the points (i), (ii) and
(iii) discussed above. 

\textbf{(i)} The conserved energy for the linear perturbation has a similar
structure as the energy of the wave equation: it is an integral over an
spacelike surface of terms that involves squares of first derivatives of the
perturbations. This energy is related with the second order expansion of the
ADM mass. However it is important to stress that the positiveness of this
energy can not be easily deduced from the positiveness of the ADM mass. In
fact, as we will see, this result is proved as a consequence of highly
non-trivial identities.  It is also important to emphasize that this energy is
positive also inside the ergosphere.

The energy expression and its conservation do not require any mode expansion of
the fields. The existence of this conserved quantity provides a basic non-modal
stability criteria for axially symmetric linear perturbation of the extreme
Kerr black hole. Since the equation are linear and the coefficients of them do
not depend on time, it is possible to construct an infinitely number of higher
order conserved energies. We expect that these higher order energies can be
used to prove pointwise boundedness of the solution, in a similar fashion as in
\cite{Dain:2012qw}. In that reference the pointwise boundedness of solutions of
the wave equation on the extreme Reissner-Nordstr\"om black hole was proved
using only higher order energies estimates. But, up to now, we were not able to
extend this result to the present context.  However, in the particular case of
the Minkowski background we prove a pointwise bound for the linear
perturbations in a remarkable simple way. Comparing with the Minkowski case,
the main difficulties to obtain pointwise estimates from the energy in the Kerr
case are two: first, the equations for the norm and the twist are coupled and
hence it is not possible to separate them as in the Minkowski case. Second, the
coefficient of the equations are singular at the horizon and hence we can not
use standard Sobolev estimates.

This conserved energy is closely related with the energy studied by Hollands and
Wald \cite{Hollands:2012sf} (see also \cite{Keir:2013jga}). We expect that the
techniques used here to prove positiveness should also be useful in that
context. Also, the boundary conditions at the horizon proposed in
\cite{Hollands:2012sf} are likely to be useful to generalize our results to the
non-extreme case.

\textbf{(ii)} The existence of this conserved energy and its related stability
criteria are not in contradiction with Aretakis instabilities. The situation is
very similar as the one discussed in \cite{Dain:2012qw} for the case of the
wave equation: the energy is only defined in the black hole exterior region and
it does not control any transverse derivative at the horizon.

\textbf{(iii)} As we pointed out above, the energy used here is related with
the ADM mass which is also conserved in the non-linear regime (see the
discussion in \cite{Dain:2009wi}) . That is, the energy estimates used here are
very likely to be useful in the non-linear case.

The plan of the article is the following. The expression of the conserved
energy arises naturally in a particular gauge for the Einstein equation: the
maximal--isothermal gauge. We review this gauge in section
\ref{sec:axially-symm-pert}. In that section we also present the linearized
equations on a class of stationary backgrounds. In section \ref{minkowki} we
study the particular case of the Minkowski background, where we prove that the
solutions are pointwise bounded in terms of a constant that depends only on the
conserved energy, see theorem \ref{t:minkowski}.  In section \ref{s:kerr} we
study the extreme Kerr background and we prove the main result of this article
given by theorem \ref{theorem}. Finally, in the appendices we write the Kerr
solution in the maximal--isothermal gauge and we also prove a Sobolev
like estimate needed in the proof of theorem \ref{t:minkowski}.

\section{Axisymmetric   Einstein equations  in the maximal--isothermal gauge} 
\label{sec:axially-symm-pert}

In axial symmetry, the maximal-isothermal gauge has the important property that
the total ADM mass can be written as a positive definite integral on the
spacelike hypersurfaces of the foliation and the integral is constant along the
evolution \cite{Dain:2008xr}. The conserved energy for the linear perturbations
will be obtained as an appropriate second order expansion of this integral.
In this section we first review the full Einstein equations in this gauge in
subsection \ref{sec:einstein-equations} and then in subsection
\ref{linealizacion} we perform the linearization on a class of stationary
backgrounds that include the Kerr black hole. On this class of backgrounds the
linearized equations in this gauge have a remarkably simply form.

\subsection{Einstein equations}
\label{sec:einstein-equations}
Einstein equations in the maximal-isothermal gauge were studied, with slight
variations, in several works \cite{Choptuik:2003as}, \cite{Rinne:thesis},
\cite{Garfinkle:2000hd}, \cite{Dain:2009wi}.  In this section we review these
equations, we closely follow \cite{Dain:2009wi}.

In axial symmetry, it possible to perform a symmetry reduction of Einstein
equations to obtain a set of geometrical equations in the 3-dimensional
quotient manifold in terms a Lorenzian 3-dimensional metric.  See \cite{Dain:2009wi} 
for the details.  In appendix \ref{sec:kerr-black-hole} we explicitly perform
this reduction for the Kerr metric.

On the 3-dimensional quotient manifold we take a foliation of spacelike
surfaces. The intrinsic metric on the slices of the foliation is denoted by
$q_{AB}$ and the extrinsic curvature by $\chi_{AB}$.  Here the indices
$A,B\cdots$ are 2-dimensional.

The maximal-isothermal gauge and its associated cylindrical coordinates
$(t,\rho,z)$ are defined by the following two conditions. For the the lapse,
denoted by $\alpha$, we impose the maximal condition on the 2-surfaces
$t=constant$.  That is, the trace $\chi$ of the extrinsic curvature vanishes
\begin{equation}
  \label{eq:34}
  \chi =q^{AB}\chi_{AB}=0. 
\end{equation}
The shift, denoted by $\beta^A$, is fixed by the requirement that the intrinsic
metric $q_{AB}$ has the following form
\begin{equation}
\label{eq:2 metrica}
q_{AB}=e^{2u}\delta_{AB},
\end{equation}
where $\delta_{AB}$ is the fixed flat metric  
\begin{equation}
 \label{eq:metrica plana}
 \delta=d\rho^{2}+dz^{2}.
\end{equation}

For our purposes, the relevant geometries for the 2-dimensional spacelike
surfaces are the half plane $\Rdm$ (defined by $-\infty <z<\infty$, $0\leq \rho
<\infty$) for the Minkowski case or $\Rdm\setminus \{0\}$ for the black hole
case. In that case the origin will represent an extra asymptotic end.  For both
cases the axis of symmetry is defined by $\rho=0$.

The dynamical degree of freedom of the gravitational field are encoded in two
geometrical scalars $\eta$ and $\omega$, the square of the norm and the twist
of the axial Killing vector respectively.  Due to the behaviour at the axis,
instead of $\eta$, $\alpha$ and $u$ it is often convenient to work with the
auxiliary function $\sigma$, $\bar \alpha$ and $q$ defined by
\begin{equation}
\label{eq:eta-sigma}
\eta=\rho^{2}e^{\sigma}, \quad \alpha=\rho\bar{\alpha},\quad 
 u=\ln\rho+\sigma+q. 
\end{equation}

To write the equations we will make use of the following  differential operators. The 2-dimensional Laplacian $\Ld$ defined by 
\begin{equation}
  \label{eq:Ld}
\Ld q = \partial^2_\rho q +  \partial^2_z  q,
\end{equation}
and the operator  $\Ldt$  defined as 
\begin{equation}
  \label{eq:Ldt}
  \Ldt \sigma =\Ld\sigma +\frac{\partial_\rho \sigma}{\rho}.
\end{equation}
This operator, which appears frequently in the rest of the article, is the flat
Laplace operator in 3-dimensions written in cylindrical coordinates and acting
on axially symmetric functions. The conformal Killing operator $\ck$ acting on
a vector $\beta_A$ is defined by
\begin{equation}
  \label{eq:operador killing}
  (\ck\beta)_{AB}=\partial_{A}\beta_{B}+\partial_{B}\beta_{A}-\delta_{AB}\partial_{C}\beta^{C}.
\end{equation}
In these equations $\partial$ denotes partial derivatives with respect to the
space coordinates $(\rho,z)$ and all the indices are moved with the flat metric
$\delta_{AB}$. We denote by a dot the partial derivative with respect to $t$
and we define the prime operator as
\begin{equation}
  \label{eq:prime}
\eta' =  \frac{1}{\alpha} \left(\dot \eta- \beta^A\partial_A \eta\right). 
\end{equation}

Einstein equations in the maximal-isothermal gauge are divided into three
groups: evolution equations, constraint equations and gauge equations.  The
evolution equations are further divided into two groups, evolution equations
for the dynamical degree of freedom $(\sigma, \omega)$ and evolution equations
for the metric $q_{AB}$ and second fundamental form $\chi_{AB}$. Due to the
axial symmetry, these equations are not independent (see the discussion in
\cite{Rinne:thesis}). For example, the constraint equations are essentially
equivalent to the evolution equations for the metric and second fundamental
form.  In particular, in this article we will not make use of the evolution
equations for the metric and second fundamental form, we will always use
instead a time derivative of the constraint equations.

Bellow we write the equations, for the deduction of them see \cite{Dain:2009wi}. We
divide  them in the three groups discussed above. In the next sections the
linearization of these equations on different background is performed,  for the
sake of clarity we will always group them in the same way.

\textbf{Evolution equations:}

The evolution equations for $\sigma$ and $\omega$ are given by\footnote{There
  were a misprint in equation (63) in \cite{Dain:2009wi}, a minus sign is missing
  on the right hand side of this equation. We have corrected that in equation
  (\ref{eq:evolucion-sigma}).}
\begin{align}
\label{eq:evolucion-sigma}
&-e^{2u} \sigma'' + ^{(3)}\Delta \sigma + \partial_A \sigma \frac{\partial^A \bar
   \alpha}{\bar \alpha}-2e^{2u} (\log\rho)'' +
 2\frac{\partial_\rho\bar\alpha}{\bar \alpha\rho} = 
\frac{\left(e^{2u} \omega'^2-|\partial \omega|^2 \right)}{ \eta^2},  \\ 
  \label{eq:evolucion-omega}
  &-e^{2u} \omega'' + ^{(3)}\Delta \omega + \partial_A \omega \frac{\partial^A \bar
    \alpha}{\bar \alpha} =
  \frac{2\left(\partial_A\omega \partial^A\eta-e^{2u}\omega'\eta' \right)}{\eta}. 
\end{align}
The evolution equation for the metric $q_{AB}$ (by equation (\ref{eq:2 metrica}) this is only one equation for the conformal factor $u$) and the second fundamental form $\chi_{AB}$ are given by 
\begin{align}
  \label{eq:evolucion-metrica}
  2\dot u &=\partial_A\beta^A +2 \beta^A \partial_A u,\\ 
 \label{eq:evolucion-chi}
 \dot \chi_{AB}  &= \pounds_\beta \chi_{AB}- F_{AB} -\alpha G_{AB} 
 -2\alpha \chi_{AC} \chi^{C}_{B},  
\end{align}
where $\pounds$ denotes Lie derivative and  we have defined
\begin{equation}
 \label{eq:F}
 F_{AB}=  \partial_A\partial_B\alpha -\frac{1}{2}\delta_{AB}\Delta\alpha-2\partial_{(A}\alpha \partial_{B)}u
 + \partial_C\alpha \partial^Cu \delta_{AB},
\end{equation}
and
\begin{align}
   \label{eq:G}
   G_{AB} &= {}^{(3)}R_{AB}- \frac{1}{2}\delta_{AB} \, {}^{(3)}R_{CD} \delta^{CD}, \\
 \label{eq:ricci}
 {}^{(3)}R_{AB} & = \frac{1}{2\eta^2} (\partial_A \eta\partial_B \eta + \partial_A \omega \partial_B \omega).  
\end{align}

\textbf{Constraint equations:}

The momentum and Hamiltonian constraints are given by 
\begin{align}
  \label{eq:momento}
  \partial^B \chi_{AB} &=-\frac{e^{2u}}{2\eta^2}\left(\eta' \partial_A\eta  +
   \omega' \partial_A\omega \right), \\
  \label{eq:hamiltoniano}
  ^{(3)}\Delta \sigma +\Delta q &= -\frac{\den}{4\rho},
\end{align}
where we have defined the energy density $\den$ by 
\begin{equation}
 \label{eq:epsilon}
 \den =  \left( \frac{e^{2u}}{\eta^2}\left(\eta'^2 + \omega'^2 \right)+ |\partial
 \sigma|^2 + \frac{|\partial \omega|^2}{\eta^2} +  2 e^{-2u}  \chi^{AB} \chi_{AB}\right)\rho. 
\end{equation}
It is important to emphasize that $\den$ is positive definite. 

 \textbf{Gauge equations:} 

The gauge equations for lapse and shift are given by
\begin{align}
   \label{eq:alpha}
   \Delta \alpha& = \alpha  \left (e^{-2u} \chi^{AB}  \chi_{AB} +e^{2u} \bar \mu \right ), \\ 
\label{eq:beta}
(\ck\beta )_{AB} & =   2\alpha e^{-2u}  \chi_{AB},
\end{align}
where we have defined $\bar \mu$  by
\begin{equation}
  \label{eq:36}
  \bar \mu = \frac{1}{2\eta^2}\left(\eta'^2 + \omega'^2 \right).
\end{equation}

As we mentioned above, the most important property of this gauge is that the
total ADM mass of the spacetime is given by the following integral on the half
plane $\mathbb{R}_{+}^{2}$ of the positive definite energy density $\den$
\begin{equation} 
  \label{eq:masa}
  m= \frac{1}{16}\int_{\mathbb{R}_{+}^{2}}   \den  \, d\rho dz.   
\end{equation}
Moreover, this quantity is conserved along the evolution in this gauge (see
\cite{Dain:2008xr}).  We emphasize that the domain of integration in
\eqref{eq:masa} is $\mathbb{R}_{+}^{2}$ even in the case of a black hole (see
the discussion in \cite{Dain06c}).

We have introduced two slight changes of notation with respect to
\cite{Dain:2009wi}. First we have suppressed the hat symbol over tensors like
$\hat \chi^{AB}$ introduced in \cite{Dain:2009wi} to distinguish between
indices moved with the flat metric $\delta_{AB}$ and with the metric
$q_{AB}$. In this article there is no danger of confusion since all the indices
are moved with the flat metric $\delta_{AB}$. Second, we have defined the
energy density $\den$ in (\ref{eq:epsilon}) with an extra factor $\rho$. This
is convenient for the calculations presented in the next section since the
integral in the mass (\ref{eq:masa}) has then the flat volume element in $\Rdm$
(in \cite{Dain:2009wi} the $\rho$ factor appears in the volume element). The
only disadvantage of this notation is that in the right hand side of the
Hamiltonian constraint (\ref{eq:hamiltoniano}) an extra $\rho$ appears in the
denominator.

\textbf{Boundary conditions:}

At spacelike infinity we assume the following standard asymptotically flat fall
off condition in the limit $r\to\infty$
\begin{equation}
  \label{eq:condicionesinfinito}
   \sigma, \beta^{A}, \chi_{AB}, \dot\sigma, \dot \beta^{A}, \dot \chi_{AB} =o_{1}(r^{-1/2}),\quad \bar\alpha -1=o_1(1),
 \end{equation}
 where we write $f=o_j(r^{k})$ if $f$ satisfies $\partial^\alpha
 f=o(r^{k-|\alpha|})$, for $|\alpha|\leq j$, where $\alpha$ is a multi-index
 and the spherical radius $r$ is defined by $r=\sqrt{\rho^2+z^2}$. In the
 following we will also make use of a similar notation for $f=O_j(r^k)$.

At the axis the functions must satisfy  the following parity conditions 
\begin{equation}
 \label{eq:condicionesparidad1}
 \eta,\;\omega,\;\bar \alpha,\;u,\;q,\;\sigma,\;\chi_{\rho\rho},\;\beta^{z}\; \text{are even functions of}\;\rho,
\end{equation}
and
\begin{equation}
 \label{eq:condicionesparidad2}
 \alpha, \;\chi_{\rho z},\;\beta^{\rho}\;\text{are odd functions of}\;\rho.
\end{equation}
Note that odd functions vanish at the axis and $\rho$ derivative of even functions vanishes at the axis.

In the case of extreme Kerr black hole we have an extra asymptotic end, which in these coordinates is located at  the origin. For that case we will assume the following behaviour  in the limit $r\to 0$ 
\begin{equation}
  \label{eq:condicionesend}
  \sigma, \beta^{A}, \chi_{AB}, \dot\sigma, \dot \beta^{A}, \dot \chi_{AB} =o_{1}(r^{-1/2}), \quad \bar\alpha -1=o_1(1). 
 \end{equation}
 These conditions encompass the asymptotically cylindrical behaviour typical of
 extreme black hole at this end (see the discussion in \cite{Dain06c} and
 \cite{Dain:2010uh}).

The behaviour of the twist $\omega$ is more subtle because it contains the
information of the angular momentum. It will be discussed in the
next sections.

\subsection{Linearization}
\label{linealizacion}
Denote by $\psi$ any of the unknowns of the previous equations.  Consider a
one-parameter family of exact solutions $\psi(\lambda)$. To linearize the
equations with respect to the family $\psi(\lambda)$ means to take a derivative
with respect to $\lambda$ to the equations and evaluate them at $\lambda=0$. 
We will use the following notation for the background and the first order linearization 
\begin{equation}
  \label{eq:37}
  \psi_0=\left.\psi(\lambda) \right|_{\lambda=0}, \quad \psi_1=\left . \frac{d
      \psi(\lambda)}{d\lambda }\right|_{\lambda=0}. 
\end{equation}
 
We will assume that the background solution is stationary in this gauge, that is  
\begin{equation}
  \label{eq:5}
\dot \psi_0 =0. 
\end{equation}
Moreover, we will also assume that the background shift and second fundamental form vanished
\begin{equation}
 \label{eq:beta0chi0}
 \beta^A_{0}=0,\quad \chi_{0AB}=0.
\end{equation}
The condition (\ref{eq:beta0chi0}) is satisfied by the Kerr solution for any
choice of the mass and angular momentum parameters, see appendix
\ref{sec:kerr-black-hole}. This condition simplifies considerably the
equations. In particular, from (\ref{eq:5}) and (\ref{eq:beta0chi0}) we deduce
\begin{equation}
  \label{eq:42}
 \psi'_0 =0.  
\end{equation}

The first important consequence of the background assumptions
(\ref{eq:beta0chi0}) is that the first order expansion of the lapse is
trivial. Namely, the right hand side of equation (\ref{eq:alpha}) is second
order in $\lambda$, hence we obtain
\begin{equation}
 \label{eq:alpha1}
 \Delta \alpha_0=0, \quad \Delta \alpha_{1}=0.
\end{equation}
Since the boundary condition for 
for $\alpha$ are independent of $\lambda$, it
follows that the first order perturbation $\alpha_1$ satisfies homogeneous
boundary condition both at the axis and at infinity, and hence from equation
(\ref{eq:alpha1}) we obtain that  
  \begin{equation}
    \label{eq:41}
    \alpha_1=0. 
  \end{equation}
In contrast, the zero order lapse $\alpha_0$ satisfies non-trivial boundary
conditions. The specific value of $\alpha_0$
will depend, of course, on the choice of background. Remarkably, for Minkowski
and extreme Kerr we have $\alpha_0=\rho$, as we will see in the next
sections. But for non-extreme Kerr it has a different value (see appendix
\ref{sec:kerr-black-hole}). In this section we keep $\alpha_0$ arbitrary in
order to obtain general equations that can be used in future works for
non-extreme black holes.

Using (\ref{eq:41}),  (\ref{eq:beta0chi0}) and (\ref{eq:5})  we find the following useful formulas
\begin{align}
 \label{eq:derivada-psip}
\psi'_1 &=\frac{1}{\alpha_0} \left( \dot \psi_{1}-\beta_{1}^{A}\partial_{A}\psi_{0} \right),\\
 \label{eq:derivada-psipp} 
\psi''_1  &=\dfrac{1}{\alpha^2_{0}}\left( \ddot\psi_{1}-\dot\beta_{1}^{A}\partial_{A}\psi_{0} \right).
\end{align}
Also, as consequence of the definition  (\ref{eq:eta-sigma})  we have   the following relations between  $\eta$ and $\sigma$
\begin{equation}
 \label{eq:eta0sigma0}
 \eta_{0}=\rho^{2}e^{\sigma_{0}}, \quad 
 \eta_{1}=\eta_{0}\sigma_{1}.
\end{equation}

Using these assumptions it is straightforward to obtain the linearization of
the equations presented in section \ref{sec:einstein-equations}. The result is
the following.

\textbf{Evolution equations:}

The evolution equation for $\sigma_1$ and $\omega_1$ are given by
\begin{align}
 \label{eq:evolucion-sigma2}
 -\frac{e^{2u_{0}}}{\alpha_0^2} \dot p+^{(3)}\Delta\sigma_{1}
 + \dfrac{\partial_{A}\sigma_{1}\partial^{A}\bar{\alpha_{0}}}{\bar{\alpha}_{0}} 
 &=\frac{2}{\eta_0^2}\left(\sigma_{1} \arrowvert\partial \omega_{0}\arrowvert^{2} -\partial_{A}\omega_{1}\partial^{A}\omega_{0}\right),\\
 \label{eq:evolucion-omega2}
- \frac{e^{2u_{0}}}{\alpha_0^2} \dot \po   +^{(3)}\Delta \omega_{1}
+\dfrac{\partial_{A}\omega_{1}\partial^{A}\bar{\alpha}_{0}}{\bar{\alpha}_{0}}
& =  4\dfrac{\partial_{\rho}\omega_{1}}{\rho}+ 2\partial_{A}\omega_{1}\partial^{A}\sigma_{0}+2\partial_{A}\omega_{0}\partial^{A}\sigma_{1},
\end{align}
where we have defined the following two useful auxiliary variables
\begin{align}
  \label{eq:43}
\ps &=  \dot \sigma_1 -\beta_1^A\partial_A \sigma_0 
  -2 \frac{\beta^\rho}{ \rho},\\
\po & =    \dot \omega_1 -\beta_1^A\partial_A \omega_0.
\end{align}

The evolution equation for the metric and second fundamental form are given by 
\begin{align}
 \label{eq:evolucionmetrica2}
 2\dot u_{1} &=\partial_{A}\beta_{1}^{A}+2\beta_{1}^{A}\partial_{A}u_{0},\\
 \label{eq:evolucion-chi-2}
\dot \chi_{1AB} &=-\left(F_{1AB}+\alpha_{0}G_{1AB}\right),
 \end{align}
where
\begin{equation}
 \label{eq:F1}
 F_{1AB}= -2\partial_{(A}\alpha_{0}\partial_{B)}u_{1}+\delta_{AB} \partial_{C}\alpha_{0} \partial^{C}u_{1},
\end{equation}
and
\begin{multline}
 \label{G1} G_{1AB}=\dfrac{1}{2\eta_{0}^{2}}\left(\partial_{A}\eta_{1}\partial_{B}\eta_{0}+\partial_{A}\eta_{0}\partial_{B}\eta_{1}
 +\partial_{A}\omega_{1}\partial_{B}\omega_{0}+\partial_{A}\omega_{0}\partial_{B}\omega_{1}\right)\\
 -\dfrac{\sigma_{1}}{\eta_{0}^{2}}\left(\partial_{A}\eta_{0}\partial_{B}\eta_{0}+\partial_{A}\omega_{0}\partial_{B}\omega_{0}\right)\\
 -\dfrac{\delta_{AB}}{2}\left[\dfrac{1}{\eta_{0}^{2}}\left(\partial_{C}\eta_{0}\partial^{C}\eta_{1}+\partial_{C}\omega_{0}
 \partial^{C}\omega_{1}\right)
 -\dfrac{\sigma_{1}}{\eta_{0}^{2}}\left(|\partial\eta_{0}|^2+|\partial\omega_{0}|^2\right)\right].
\end{multline}

\textbf{Constraint equations:}

The momentum constraint and Hamiltonian constraints are given by
\begin{align}
 \label{eq:momento 2}
 \partial^{B}\chi_{1AB} &=
 -\frac{e^{2u_{0}}}{2\alpha_0}  \left(  p \left(\partial_A \sigma_0 + 2\frac{\partial_A \rho}{\rho} \right)
 + \frac{ \partial_{A}\omega_{0}}{\eta_0^2}d \right),\\
 \label{hamiltoniano 2}
 ^{(3)}\Delta\sigma_{1}+\Delta q_{1} &=-\dfrac{\den_{1}}{4\rho},
\end{align}
where $\den_{1}$ is the first order term of  the energy density
(\ref{eq:epsilon}), that is 
\begin{equation}
 \label{eq:epsilon2}
 \den_1=\left( 2\partial_{A}\sigma_{0}\partial^{A}\sigma_{1} +\dfrac{2\partial_{A}\omega_{0}\partial^{A}\omega_{1}}{\eta_{0}^{2}}-\dfrac{2\sigma_{1}|\partial\omega_{0}|^2}{\eta_{0}^{2}}\right)\rho.
 \end{equation}

\textbf{Gauge equations:}

We have seen that the first order lapse is zero. For the shift we have
\begin{eqnarray}
 \label{eq:beta 2}
 \left(\ck\beta_1\right)^{AB}=2e^{-2u_{0}}\alpha_{0} \chi_{1}^{AB}.
 \end{eqnarray}

\vspace{1cm}

We have presented above the complete set of axially symmetric linear equations
in the maximal--isothermal gauge. The conserved energy for this
system of equation is calculated from the second variation of the energy
density (\ref{eq:epsilon}) as follows. Assume that $\psi(\lambda)$ has the following form
 \begin{equation}
   \label{eq:44}
   \psi(\lambda)=\psi_0 + \lambda \psi_1. 
 \end{equation}
That it, we assume that the second order derivative with respect to $\lambda$
of $\psi(\lambda)$ vanishes at $\lambda=0$. For this kind of linear
perturbations we define the second variation of $\den$ as
\begin{equation}
  \label{eq:45}
  \den_2=\left.\frac{d^2 \den(\lambda)}{d \lambda^2}\right|_{\lambda=0}.
\end{equation}
Using  (\ref{eq:epsilon}) we obtain
\begin{multline} 
 \den_{2}=\left( \dfrac{2e^{2u_{0}}}{\alpha_0^2}\left(
 p^2
 +\frac{d^2}{\eta_0^2}\right)-8\dfrac{\sigma_{1}\partial_{A}\omega_{0}\partial^{A}\omega_{1}}{\eta_{0}^{2}} +
\right.
\\
\left. +2|\partial\sigma_{1}|^2
 +4e^{-2u_{0}}{\chi}^{AB}_{1}\chi_{1AB}
 +2\dfrac{\arrowvert\partial\omega_{1}\arrowvert^{2}}
{\eta_{0}^{2}}+4\dfrac{\arrowvert\partial\omega_{0}\arrowvert^{2}}{\eta_{0}^{2}}\sigma_{1}^{2}\right)\rho.
\end{multline}
Note that $\den_2$, in contrast with $\den$, is not positive definite. 

For further reference we write also the zero order expression for the energy density
\begin{equation}
 \label{eq:epsilon0}
 \den_0 = \left(|\partial
 \sigma_{0}|^2 + \frac{|\partial \omega_{0}|^2}{\eta_{0}^2}\right)\rho, 
\end{equation}
and the masses associated with  the different orders of the energy density
\begin{align}
  \label{eq:m0}
  m_0 =\frac{1}{16}\int_{\Rdm} \den_0 \, d\rho dz,\\  
m_1 =\frac{1}{16} \int_{\Rdm} \den_1 \, d\rho dz, \label{eq:m1} \\
 m_2 =\frac{1}{16}\int_{\Rdm} \den_2 \, d\rho dz. \label{eq:m2}
\end{align}
Recall that $\den_1$ has been calculated in (\ref{eq:epsilon2}).

We will prove that $m_1$ vanished and that $m_2$ is  conserved and positive
definite. Since we are interested in the study of linear stability, it
is important for our present purpose (and also for future works on this
subject) to prove these statements using only the linear equations, without
referring to the original non-linear system. In the next sections we will perform
these proofs.  However, from the conceptual point of view and for further
possible applications to the non-linear stability problem, it is important also
to deduce these properties from the full equations. We discuss this point
bellow.

Consider a general one-parameter family of exact solutions $\psi(\lambda)$
(i.e. we are not assuming the particular linear form (\ref{eq:44})). 
For this family we compute the exact mass $m(\lambda)$ given
  by equation \eqref{eq:masa}. This quantity  is conserved,
  that is
\begin{equation}
  \label{eq:1w}
  \frac{d m(\lambda)}{dt}=0,
\end{equation}
This equation is valid for all $\lambda$. Taking derivatives with respect
to $\lambda$ of equation (\ref{eq:1w}) and then evaluating them in $\lambda=0$
we  obtain that
\begin{align}
\frac{d}{dt} \left. m  \right|_{\lambda=0}=0,  \label{eq:2a}\\
\frac{d}{dt} \left. \frac{d m}{d \lambda}\right|_{\lambda=0}= 0,  \label{eq:2b}\\
\frac{d}{dt} \left. \frac{d^2 m}{d \lambda^2}\right|_{\lambda=0}= 0.\label{eq:2c}
\end{align}
We can, of course, take more derivatives with respect to $\lambda$, but this
will not provide any useful conserved quantity for the linear equations.

It is clear that equations (\ref{eq:2a}) and (\ref{eq:2b}) are precisely
  \begin{align}
\frac{d m_0}{dt} =0, \label{eq:57a} \\
\frac{d m_1}{dt} =0,  \label{eq:57b}
\end{align}
where $m_0$ and $m_1$ are given by (\ref{eq:m0}) and (\ref{eq:m1})
respectively. 
 
The first equation (\ref{eq:57a}) asserts that the mass of the background
metric is conserved. This is of course valid even when the background solution
is not stationary. In our case, since the background metric is stationary, not
only $m_0$ is conserved but also the integrand $\varepsilon_0$, given by
equation (\ref{eq:epsilon0}), is time independent, and hence the conservation
(\ref{eq:57a}) is trivial.

Since  $m_1$  depends  only on  the background 
solution $\psi_0$ and the first order perturbation $\psi_1$ (recall that $\psi_0$ and
$\psi_1$ are defined by (\ref{eq:37}) for a general family $\psi(\lambda)$)
then equation (\ref{eq:57b}) asserts that $m_1$ is a conserved quantity
for the linear equations. That is, from the exact conservation law (\ref{eq:1w})
we have deduced the conservation of $m_1$ for the linear equations. 

For a general background, $m_1$ will be non-zero. However, using the
Hamiltonian formulation of General Relativity, it is possible to show that the
first variation of the ADM mass vanishes on stationary solutions (see
\cite{Bartnik04} and reference therein).  In section \ref{s:kerr} we explicitly
perform this computation adapted to our settings. 

For the third equation (\ref{eq:2c}) the situation is different. This equation
asserts that the quantity
\begin{equation}
  \label{eq:4mm}
  \hat m_2 = \left. \frac{d^2 m}{d \lambda^2}\right|_{\lambda=0},
\end{equation}
is conserved
\begin{equation}
  \label{eq:57c}
  \frac{d \hat m_2}{dt} =0.
\end{equation}
However, $\hat m_2$ depends on the background solution $\psi_0$, the linear perturbation $\psi_1$
but also on the second order perturbation
\begin{equation}
  \label{eq:5b}
 \psi_2 =\left.\frac{d^2 \psi(\lambda)}{d \lambda^2}\right|_{\lambda=0}.
\end{equation}
Then  $\hat m_2$ is not a quantity that can be computed purely in
terms of the background solution $\psi_0$  and the linear perturbation $\psi_1$
and hence it can not  be used for the linearized equations. 

Note that the mass $m_2$ defined in (\ref{eq:m2}) is computed only using first
order perturbations (since we have assumed (\ref{eq:44}) to compute it). In principle, $m_2$
and $\hat m_2$ are different quantities.  Hence the conservation law
\begin{equation}
  \label{eq:6mm}
  \frac{d m_2}{dt}=0,
\end{equation}
can not   be deduced directly from (\ref{eq:57c}). But, as we will prove
bellow, it turns out that if the background is stationary and hence the first
variation $m_1$ vanishes, then we have $\hat m_2=m_2$. 

Let us compute explicitly $\hat m_2$. We define 
\begin{equation}
  \label{eq:45b}
 \hat  \den_2=\left.\frac{d^2 \den(\lambda)}{d \lambda^2}\right|_{\lambda=0}.
\end{equation}
We emphasize that in (\ref{eq:45b}) we are not assuming (\ref{eq:44})
and hence this is different from  (\ref{eq:45}). The difference between
$\den_2$ and $\hat \den_2$ is given by
\begin{equation}
  \label{eq:57hat}
  \hat  \den_2-\den_2=\left( 2\partial_{A}\sigma_{0}\partial^{A}\sigma_{2}
    +\dfrac{2\partial_{A}\omega_{0}\partial^{A}\omega_{2}}{\eta_{0}^{2}}-\dfrac{2\sigma_{2}|\partial\omega_{0}|^2}{\eta_{0}^{2}}\right)\rho. 
\end{equation}
In this calculation we have assumed that the background is stationary in this
gauge (namely, we have assumed (\ref{eq:5}) and (\ref{eq:beta0chi0})). The
difference between $\den_2$ and $\hat \den_2$ involves, of course, the second
order perturbation $\sigma_2$ and $\omega_2$. However, remarkably, the right
hand side of (\ref{eq:57hat}) has exactly the same for as the first variation
$\den_1$ if we replace $\sigma_1$ and $\omega_1$ in $\den_1$ (given by
(\ref{eq:epsilon2})) by $\sigma_2$ and $\omega_2$.  Hence, if $m_1$ vanishes on
stationary solutions then $\hat m_2=m_2$ (that is, the integral of the right
hand side of (\ref{eq:57hat}) vanishes).  In fact, this result is general
and well known in the calculus of variations with non-linear variations (see,
for example, \cite{Giaquinta96} p. 267).

Finally, let us discuss the sign of the second variation $m_2$.  On Minkowski,
the positive mass theorem clearly implies that the second variation of the mass
should be positive since flat space is a global minimum of the mass.  In the
extreme Kerr case there is no obvious connection between the positivity of the
mass and the second variation. However, it has been proved that the mass has a
minimum at extreme Kerr under variations with fixed angular momentum
\cite{Dain05d}\cite{Dain06c}. To prove the positivity of the second variation
$m_2$ on extreme Kerr in section \ref{s:kerr} we will use similar techniques as
in those references.  As we pointed out above, for our purpose, it is
important to prove this in terms only of the linearized equations.

\section{Minkowski perturbations}
\label{minkowki}

The natural first application of the linear equations obtained in section
\ref{linealizacion} is to study the linear stability of Minkowski in axial
symmetry. The problem of linear stability of Minkowski, without any symmetry
assumptions, was solved in \cite{christodoulou90} and the non-linear
stability of Minkowski was finally proved in \cite{Christodoulou93}. The
purpose of this section in to provide an alternative proof of the linear
stability of Minkowski in axial symmetry using the gauge presented in the
previous section. This is given in theorem \ref{t:minkowski} which constitutes
the main result of this section. 

 In comparison with the results in \cite{christodoulou90}, theorem
 \ref{t:minkowski} has the obvious disadvantage that it only applies to
axially symmetric perturbation. Moreover in this theorem only pointwise boundedness
of the solution is proved and not precise  decay rates as in
\cite{christodoulou90}. However, the advantage of this result is that it make
use only of energy estimates that can be generalized to the black hole case as
we will see in section \ref{s:kerr}. 

This system of linear equations was studied numerically in \cite{Dain:2009wi} and
analytically in \cite{Dain:2010ne}. The main difficulty is that the system is
formally singular at the axis where $\rho=0$. Theorem \ref{t:minkowski}
generalize those works by including the twist and, more important, by obtaining
a pointwise estimate of the solution in terms of conserved energies. We explain
in more detail this point bellow.

The  Minkowski background satisfies  the assumptions (\ref{eq:beta0chi0}). The
value of the other background quantities are the following 
\begin{equation}
  \label{background0}
  \omega_{0} = q_{0}= \sigma_{0}=0, 
\end{equation}
and 
\begin{equation}
  \label{background1}
 u_{0}=\ln\rho,\quad  \eta_{0} = \rho^{2}, \quad \alpha_{0}=\rho.
\end{equation}
Introducing the background quantities (\ref{background0})--(\ref{background1})
on the linearized equations obtained in section \ref{linealizacion} we arrive
at the following set of equations for the linear axially symmetric
perturbations of Minkowski.

\textbf{Evolution equations:}

The evolution equations for $\sigma_{1}$ and $\omega_{1}$ are given by
\begin{align}
 \label{eq:evo-sigma-mini}
 - \dot p + \Ldt\sigma_{1} &=0,\\
 \label{eq:evo-omega-min}
 -\ddot{\omega}_{1} +\Ldt \omega_{1}&=4\dfrac{\partial_{\rho}\omega_{1}}{\rho},
\end{align}
where we defined the auxiliary function $p$ by
\begin{equation}
 \label{eq:p-minkowski}
 p=\dot{\sigma}_{1}-\dfrac{2\beta_{1}^{\rho}}{\rho}.  
\end{equation}

The evolution equations for the metric and the extrinsic curvature are given by
\begin{align}
\label{eq:evolucion-u-minkowski} 
 2\dot{u}_{1}&=\partial_{A}\beta^{A}_{1}+2\dfrac{\beta^{\rho}_{1}}{\rho},\\
 \label{eq:evolucion-chi-minkowski}
 \dot{\chi}_{1AB} &=2\partial_{(A}q_1\partial_{B)}\rho-\delta_{AB}\partial_{\rho}q_1.
\end{align}

\textbf{Constraint equations:}

The momentum and the Hamiltonian constraints takes the following form
\begin{align}
 \label{eq:momento-minkowski}
 \partial^{A}\chi_{1AB} &=-p\partial_{B}\rho, \\
 \label{eq:hamiltoniano minkowski}
 \Delta q_{1}+^{(3)}\Delta\sigma_{1} &=0.
\end{align}

\textbf{Gauge equations for lapse and shift:}

We have proved in section \ref{linealizacion} that the first order lapse is
zero.  The equation for the shift is given by
\begin{equation}
 \label{eq:beta-minkowski}
 (\ck\beta_{1})^{AB}=\dfrac{2}{\rho} \chi_{1}^{AB}. 
\end{equation}

For the mass density we have that
\begin{equation}
  \label{eq:4}
  \den_0=\den_1=0,
\end{equation}
and hence we have
\begin{equation}
  \label{eq:61}
   m_0=m_1=0.
\end{equation}
The second order mass density is given by
\begin{equation}
\label{eq:density-flat}
 \den_{2}=\left(2p^{2}+2\dfrac{\dot
     {\omega}_{1}^{2}}{\rho^{4}}+2\arrowvert\partial\sigma_{1}\arrowvert^{2}+2\dfrac{\arrowvert\partial 
   \omega_{1}\arrowvert^{2}}{\rho^{4}}+4\dfrac
 {\chi_{1}^{AB}\chi_{1AB}}{\rho^{2}}\right) \rho. 
\end{equation}
It is important to note that $\den_2$, in the particular case of the Minkowski
background, is positive definite.
 
Before presenting the main result, let us first discuss two simple but
important properties of this set of equations. The first one (which only holds
for the Minkowski background) is that the equation for the twist $\omega_1$
(\ref{eq:evo-omega-min}) decouples completely from the other equations
\footnote{We thank O. Rinne for pointing this out to us before this work was
  started.}. Then, it is useful to split the density $\varepsilon_2$ in two
terms
\begin{equation}
  \label{eq:6}
  \varepsilon_{2}=\es+\eo,
\end{equation}
where
\begin{align}
  \label{eq:7}
  \es & =\left(2p^{2}+2\arrowvert\partial\sigma_{1}\arrowvert^{2}+4\dfrac{
   \chi_{1}^{AB}\chi_{1AB}}{\rho^{2}}\right)\rho,\\ 
\eo &=  2\dfrac{\dot
   {\omega}_{1}^{2}}{\rho^{3}} +2\dfrac{\arrowvert\partial
   \omega_{1}\arrowvert^{2}}{\rho^{3}},
\end{align}
and the corresponding masses
\begin{equation}
  \label{eq:8}
  m_2=\ms +\mo, 
\end{equation}
where
\begin{equation}
  \label{eq:9}
 \ms=\int_{\mathbb{R}_{+}^{2}} \es  \, d\rho dz,
 \quad  \mo =\int_{\mathbb{R}_{+}^{2}} \eo \, d\rho dz.
\end{equation}
Note that all the densities are positive definite.

Equation (\ref{eq:evo-omega-min}) is equivalent to the following homogeneous
wave equation
\begin{equation}
  \label{eq:38}
  -\ddot{\bar{\omega}}_1 + ^{(7)} \Delta \bar \omega_1=0,
\end{equation}
where $^{(7)} \Delta$ is the Laplacian in 7-dimensions acting on axially
symmetric functions \footnote{The trick of writing the 2-dimensional equations
  that appears in axially symmetric (which are formally singular at the axis)
  as regular equations in higher dimensions has provided to be very useful. It
  has been used, in a similar context, in \cite{Weinstein90} and
  \cite{Andreasson:2012cq}.}, namely
\begin{equation}
 \label{eq:9 minkowski}
 ^{(7)}\Delta 
\bar \omega_1 =\Delta\bar \omega_1 +5\dfrac{\partial_{\rho} \bar\omega_1}{\rho},
\end{equation}
and we have defined
\begin{equation}
  \label{eq:39}
  \bar \omega_1 =\frac{\omega_1}{\rho^4}. 
\end{equation}
That is, the dynamic of the twist potential is determined by a wave equation
and hence it is clear how to obtain decay estimates for the solution. In
contrast, the equations for $\sigma_1$ are coupled and non-standard due to the
formal singular behaviour at the axis (see the discussion in \cite{Dain:2009wi}
and \cite{Dain:2010ne}).  The wave equation (\ref{eq:38}) has associated the
canonical energy density
\begin{equation}
  \label{eq:40}
  \epsilon_{\bar\omega}= 2\left(  \dot{\bar{\omega}}^2_1 +|\partial
      \bar{\omega}_1|^2 \right)\rho^5,
\end{equation}
and corresponding energy
\begin{equation}
  \label{eq:48}
  m_{\bar\omega}=\int_{\Rt} \epsilon_{\bar\omega} \, d\rho dz.
\end{equation}
The factor $\rho^5$ in (\ref{eq:40}) comes from the expression of the volume
element in 7-dimensions in terms of the cylindrical coordinates $dx^7=\rho^5
d\rho dz$. The two densities $\den_{\bar\omega}$ and $\eo$ are related by a
boundary term
\begin{equation}
  \label{eq:49}
  \epsilon_{\bar\omega}-\eo=-4 \partial_\rho \left(\frac{\omega^2_1}{\rho^4}\right),
\end{equation}
and hence $m_{\bar\omega}=\mo$ provided $\omega_1$ satisfies appropriate
boundary conditions. Note that equation (\ref{eq:38}) suggests that
$\bar \omega_1$ and not $\omega_1$ is the most convenient variable to impose
the boundary conditions.

The second property (which will be also satisfied for the Kerr background and
in general for any stationary background) is the following. The coefficients of
the equations do not depend on time, hence if we take a time derivative to all
equations we get a new set of equations for the time derivatives of the
unknowns which are formally identical to the original ones. That is, the
variables $\sigma_1,\omega_1,u_1, \beta_1,\chi_1$ satisfy the same equations as
the time derivatives $\dot \sigma_1,\dot \omega_1,\dot u_1, \dot \beta_1,\dot
\chi_1$.  And the same is of course true for any number of time derivatives. In
particular, if $m$ is a conserved quantity, then we automatically get an
infinity number of conserved quantities which has the same form as $m$ but in
terms of the $n$-th time derivatives of $\sigma_1,\omega_1,u_1,
\beta_1,\chi_1$. For example, let us consider the mass $\ms$ defined by
(\ref{eq:7}) and (\ref{eq:9}). It depends on the the functions $p$, $\sigma_1$
and $\chi_1$, to emphasize this dependence we use the notation
$\ms[p,\sigma_1,\chi_1]$. Then we define $\ms[\dot p,\dot \sigma_1,\dot
\chi_1]$ as
\begin{equation}
  \label{eq:32}
\ms[\dot p,\dot \sigma_1,\dot \chi_1]=\int_{\mathbb{R}_{+}^{2}}  \left( 2\dot
  p^{2}+2\arrowvert\partial\dot \sigma_{1}\arrowvert^{2}+4\dfrac{
  \dot  \chi_{1}^{AB}\dot \chi_{1AB}}{\rho^{2}} \right)  \rho d\rho dz.  
\end{equation}
If $\ms[p,\sigma_1,\chi_1]$ is conserved along the evolution then also
$\ms[\dot p,\dot \sigma_1,\dot \chi_1]$ is conserved. 
The same applies for $\mo[\omega_1]$ and $m_{\bar \omega}[\bar \omega_1]$, for
example we have 
\begin{equation}
  \label{eq:35}
  m_{{\bar \omega}} [\dot{\bar \omega}_1] =\int_{\mathbb{R}_{+}^{2}} \left(\ddot{\bar{\omega}}^2_1 +|\partial
    \dot{\bar{\omega}}_1|^2 \right)\,\rho^5 d\rho dz.
\end{equation}
We will make use also of the higher order masses
$m_{\bar\omega}[\ddot{\bar\omega}_1]$ and
$m_{\bar\omega}[\dddot{\bar\omega}_1]$.

\begin{theorem}
 \label{t:minkowski}
 Consider a smooth solution of the linearized equations presented above that
 satisfies the fall off conditions at infinity (\ref{eq:condicionesinfinito})
 and the regularity conditions at the axis (\ref{eq:condicionesparidad1}),
 (\ref{eq:condicionesparidad2}). Assume also that
\begin{equation}
  \label{eq:65}
 \dot{\bar\omega}_1, \bar\omega_1= O_1(1),  
\end{equation}
at the axis and
\begin{equation}
  \label{eq:75}
  \dot{\bar\omega}_1, \bar\omega_1= o_1(r^{-5/2}), 
\end{equation}
at infinity, where we have defined
\begin{equation}
  \label{eq:76infM}
  \bar\omega_1=\frac{\omega_1}{\rho^4}.
\end{equation}

 Then, we have:
\begin{itemize}
\item[(i)] The masses $\ms$, $\mo$ and $m_{\bar \omega}$ defined by
  (\ref{eq:9}) and (\ref{eq:48}) are conserved along de evolution and $\mo=m_{\bar
    \omega}$. And hence, all higher order masses are also conserved. 

\item[(ii)] The solution $\sigma_1, \omega_1$ satisfy the following (time
  independent) bounds   
  \begin{align}
    \label{eq:1}
   C |\sigma_1|&\leq \ms[p,\sigma_1,\chi_1]+ \ms[\dot p,\dot \sigma_1,\dot \chi_1], \\
 C\frac{|\omega_1| }{\rho^4} &\leq  \mo[\ddot \omega_1]+\mo[\dddot \omega_1] , \label{eq:1b}
  \end{align}
where $C>0$ is a numerical constant. 
\end{itemize}

\end{theorem}
The value of $\omega$ at the axis determines the angular momentum (see, for
example, \cite{Dain06c}). Hence, the physical interpretation of the boundary
conditions (\ref{eq:65}) is that the perturbations do not change the angular
momentum of the background (which is zero in the case of Minkowski).   

The conservation of $\ms$ in point (i) was proved in \cite{Dain:2009wi}  . For
completeness we review this proof and also we perform it in different variables
which are the appropriate ones for the extreme Kerr black hole case treated in
the next section.

We have already shown that the equation for $\omega_1$ is decoupled and it can
be converted into an standard wave equation in higher dimensions. Hence the
dynamics of $\omega_1$ is well known. In particular one has the classical
pointwise estimates for solutions of the wave equation in 7-dimensions $|\bar \omega_1| \leq
t^{-3} C$, where the constant $C$ depends only on the initial data (see, for
example, \cite{shatah98}). We present the weaker estimate (\ref{eq:1b}) because
it can be proved using only the conserved energies and is likely to be useful
in the more complex case of the Kerr black hole, where the pure wave equations
estimates are not available.

The most important part of theorem \ref{t:minkowski} is the estimate
(\ref{eq:1}). In a previous work \cite{Dain:2010ne} the existence of solution
of this set of equations was proved using an explicit (but rather complicated)
representation in terms of integral transforms. In contrast, the a priori
estimate (\ref{eq:1}) is proved in terms of only the conserved masses in a
remarkably simple way.  This estimate is expected to be useful in the
non-linear regime.

\begin{proof}

(i) Since the equations are decoupled, we can treat the conservation for $\ms$
and $\mo$ separately. We begin with $\ms$. Taking the time derivative of $\es$ we obtain
\begin{equation}
\label{eq:m-sub-sigma-minkowski}   
\dot \den_\sigma= 4\rho p\dot{p}+4\rho \partial_{A}\sigma_{1}\partial^{A}\dot{\sigma}_{1}+8\dfrac{
    \chi_{1}^{AB}\dot{\chi}_{1AB}}{\rho}.
\end{equation}
The strategy is to prove (using the linearized equations) that the right hand
side of (\ref{eq:m-sub-sigma-minkowski}) is a total divergence and hence it
integrates to zero (under appropriate boundary conditions).  We calculate each
terms individually.

For the first term we just use the definition of $p$
given in equation (\ref{eq:p-minkowski}) to obtain
\begin{equation} 
  4\rho p \dot p = 4\rho \dot \sigma_1 \dot p -8\beta_1^\rho \dot p . \label{eq:11b}
\end{equation}

For the second term we obtain
\begin{align}
  \label{eq:12}
  4\rho \partial_{A}\sigma_{1}\partial^{A}\dot{\sigma}_{1} &=4\partial^A\left(\rho
    \dot \sigma_1 \partial_A \sigma_1\right)- 4 \dot \sigma_1 \partial^A
  \left(\rho \partial_A\sigma_1  \right),\\
&=4\partial^A\left(\rho
    \dot \sigma_1 \partial_A \sigma_1\right)-4 \rho \dot \sigma_1\,  \Ldt
  \sigma_1,\label{eq:12b}\\
&=4\partial^A\left(\rho
    \dot \sigma_1 \partial_A \sigma_1\right)-4 \rho \dot \sigma_1  \dot{p},\label{eq:12c}
\end{align}
where in line \eqref{eq:12b} we have used the definition of the operator $\Ldt$
given in equation \eqref{eq:Ldt} and in line \eqref{eq:12c} we have used equation  (\ref{eq:evo-sigma-mini}).

Finally, for the third term we have
\begin{align}
  \label{eq:13}
  8\dfrac{\chi_{1}^{AB}\dot{\chi}_{1AB}}{\rho} &=4
  (\mathcal{L}\beta_1)^{AB}\dot{\chi}_{1AB},\\
&= 8  \partial^A \beta_1^{B}\dot{\chi}_{1AB},\label{eq:13b}\\
&= 8\partial^A\left( \beta_1^{B}\dot{\chi}_{1AB}  \right)-8
\beta_1^{B} \partial^A \dot{\chi}_{1AB},\label{eq:13c}\\
&= 8\partial^A\left( \beta_1^{B}\dot{\chi}_{1AB}  \right)+8 \dot p \beta_1^\rho,
\label{eq:13d}
\end{align}
where in line (\ref{eq:13}) we have used  the gauge equation
(\ref{eq:beta-minkowski}), in line (\ref{eq:13b}) the fact that $\chi_{1AB}$ is
trace-free and in line (\ref{eq:13c}) we have used the time derivative of
equation (\ref{eq:momento-minkowski}). 

Summing these results we see that only the total divergence  terms remain. That is
\begin{equation}
  \label{eq:14}
    \dot{\varepsilon}_\sigma=\partial_A t^A,
\end{equation}
where
\begin{equation}
  \label{eq:15}
  t_A=4\rho \dot \sigma_1 \partial_A \sigma_1 + 8 \beta_1^B \dot \chi_{1AB}.
\end{equation}
We integrate (\ref{eq:14}) in the half disk $D_L$ of radius $L$ in $\Rdm$,
where $C_L$ denote the semi-circle of radius $L$, see figure \ref{fig:1}.
\begin{figure}
  \centering
\includegraphics{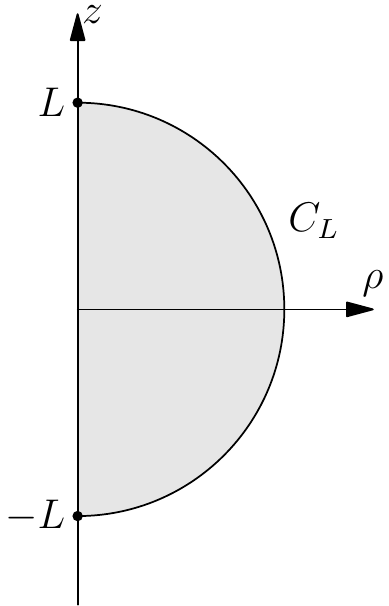}
  \caption{Domain of integration in $\Rdm$.}
  \label{fig:1}
\end{figure}
Using the divergence theorem in 2-dimensions we obtain
\begin{align}
  \label{eq:24}
  \int_{D_L} \dot \den_\sigma\, d\rho dz &= \int_{D_L} \partial_A t^A \,
  d\rho dz,\\
&= \int_{\partial D_L} t^An_A \, ds,\\
&= -\int_{-L}^L t^\rho|_{\rho=0} \, dz +\int_{C_L} t^An_A \, ds. \label{eq:24b}
\end{align}
where $n^A$ is the outwards unit normal vector and $ds$ the line element of $C_L$. 

The integrated of the first term in line \eqref{eq:24b} is given by
\begin{equation}
  \label{eq:25}
  t_\rho= 4\rho \dot \sigma_1 \partial_\rho \sigma_1 + 8 \beta_1^\rho \dot{\chi}_{1\rho \rho}+8 \beta_1^z \dot{\chi}_{1\rho z}.
\end{equation}
The first term clearly vanished at the axis $\rho=0$. The second and third term
also vanish  at the axis since the regularity  conditions  (\ref{eq:condicionesparidad2})  implies that
$\beta_1^\rho$ and $\chi_{1\rho z}$ are zero at the axis.  Hence we obtain
\begin{equation}
  \label{eq:26}
 \int_{D_L}  \dot \den_\sigma \, d\rho dz =\int_{C_L} t^An_A \, ds. 
\end{equation}
Taking the limit $L\to \infty$ and using the fall off conditions
 (\ref{eq:condicionesinfinito})  we obtain that the integral vanished and hence
$\dot m_\sigma =0$ (recall that on $C_L$ we have $ds=r d\theta$ where
$\tan\theta=z/\rho$).

The conservation of $\mo$ is similar. We take the time derivative of the mass
density $\eo$
\begin{equation}
 \label{eq:epsilonsubomegaminkowski}
 \dot  \den_\omega =4\dfrac{\dot \omega_{1}\ddot
   \omega_1}{\rho^{3}}+4\dfrac{\partial_A\omega_{1}\partial^A \dot \omega_1}{\rho^{3}}.
\end{equation}
For the first term we have
\begin{equation}
  \label{eq:28}
 4\dfrac{\dot \omega_{1}\ddot
   \omega_1}{\rho^{3}}=4\dfrac{\dot
   \omega_{1}}{\rho^{3}}\left( \partial^A\partial_A \omega_1 - \frac{3\partial_\rho \omega_1}{\rho}  \right) ,
\end{equation}
where we have used equation (\ref{eq:evo-omega-min}).

For the second term we have
\begin{equation}
  \label{eq:29}
  4\dfrac{\partial_A\omega_{1}\partial^A \dot
    \omega_1}{\rho^{3}}=4 \partial^A\left( \frac{\dot\omega_1 \partial_A
      \omega_1}{\rho^3}\right)- 4\frac{\dot \omega_1}{\rho^3} \left( \partial^A\partial_A \omega_1 - \frac{3\partial_\rho \omega_1}{\rho}  \right) .
\end{equation}
Hence we obtain
\begin{equation}
  \label{eq:30}
  \dot   \den_\omega  =\partial_A t^A,
\end{equation}
with
\begin{equation}
  \label{eq:31}
  t_A= 4 \frac{\dot\omega_1 \partial_A
      \omega_1}{\rho^3}.
\end{equation}
Integrating in the same domain as above and using the behavior at the axis
(\ref{eq:65}) and the fall off conditions  (\ref{eq:75})  at
infinity we obtain that $\dot{m}_\omega=0$.  Finally, the equality
${m}_\omega={m}_{\bar\omega}$ is deduced from (\ref{eq:49}) and the assumption
(\ref{eq:65}).

(ii) To prove the estimate   (\ref{eq:1}) note that we have the following bounds
\begin{align}
  \label{eq:10}
\ms[p,\sigma_1,\chi_1] &\geq \int_{\Rdm} |\partial \sigma_1|^2 \rho \, d\rho dz, \\
\ms[\dot p,\dot \sigma_1,\dot \chi_1] & \geq 2\int_{\Rdm} \dot p^2 \rho \, d\rho
dz =2\int_{\Rdm} \left( \Ldt \sigma_1\right)^2 \rho \, d\rho dz, \label{eq:10b}
\end{align}
where in the last equality of line \eqref{eq:10b} we have used equation
(\ref{eq:evo-sigma-mini}). The right hand side of \eqref{eq:10b} can be written
in the following form
\begin{align}
  \label{eq:59}
  \int_{\Rdm} \left( \Ldt \sigma_1\right)^2 \rho \, d\rho dz &=
  \int_{\mathbb{R}^3} (\Ldt \sigma_1)^2  \, dx^3,  \\
& = \int_{\mathbb{R}^3} | \partial^2\sigma_1|^2 \, dx^3.\label{eq:59d}
\end{align}
where in the right hand side of line \eqref{eq:59} we have changed from
cylindrical coordinates $(\rho,z)$ to Cartesian coordinates $(x,y,z)$ in $\Rt$,
with $x=\rho \cos\phi$, $y=\rho \sin \phi$. For axially symmetric functions
(i.e. functions in $\Rt$ that do not depends on $\phi$) we have that $ dx^3 =
\rho \, d\rho dz$. In Cartesian coordinates the Laplacian $\Ldt$ is given by
\begin{equation}
  \label{eq:57laplacian}
  \Ldt \sigma_1= \partial^2_x \sigma_1 + \partial^2_y \sigma_1 + \partial^2_z \sigma_1.
\end{equation}
And in line \eqref{eq:59d} we have integrated by parts, due
to the fall off assumptions on $\sigma_1$ the boundary terms vanishes. In this
equation $|\partial^2\sigma_1|^2$ denote the sum of the squares of all
second derivatives in terms of the Cartesian coordinates in $\Rt$, that is
\begin{equation}
  \label{eq:74}
  |\partial^2\sigma_1|^2= (\partial^2_x \sigma_1)^2 + (\partial^2_y \sigma_1)^2
  + (\partial^2_z \sigma_1)^2 + (\partial_x\partial_y \sigma_1)^2
  +(\partial_x\partial_z\sigma_1)^2 +(\partial_y\partial_z \sigma_1)^2.
\end{equation}

From (\ref{eq:59d}), (\ref{eq:10b}) and (\ref{eq:10}) we obtain the following
crucial estimate
\begin{equation}
  \label{eq:33}
  \ms[p,\sigma_1,\chi_1]+ \ms[\dot p,\dot \sigma_1,\dot \chi_1] \geq \int_{\mathbb{R}^3} 
 \left( | \partial^2\sigma_1|^2  +  |\partial \sigma_1|^2\right)  \, dx^3.
\end{equation}
Note that on the right hand side of (\ref{eq:33}) there are no terms with
$\sigma^2_1$ and hence we can not use directly the standard Sobolev estimate to control
pointwise the solution $\sigma_1$. However, using the  estimate given by lemma
\ref{l:estimaten} with $n=3$ and $k=2$ we  obtain the desired
result  (\ref{eq:1}).

To obtain the estimate (\ref{eq:1b}) for $\bar \omega_1$, we proceed in a similar
manner. From the definition of $m_{\bar\omega}$ we obtain
\begin{align}
  \label{eq:11}
 m_{\bar\omega}\geq \int_{\Rdm} |\partial \bar\omega_1|^2 \rho^5 \, d\rho dz=
 \int_{\mathbb{R}^7}  |\partial \bar\omega_1|^2 \, dx^7,
\end{align}
where we have used that $dx^7=\rho^5 d\rho dz$. For the higher order masses we have
\begin{align}
  \label{eq:50}
  m_{\bar\omega}[\dot{\bar{\omega}}_1] &\geq \int_{\Rdm} \ddot{\bar\omega}^2_1 \rho^5 \, d\rho dz,\\
& = \int_{\mathbb{R}^7} ( {}^{(7)}\Delta \bar\omega_1)^2  \, dx^7, \label{eq:50b}  \\
& = \int_{\mathbb{R}^7} | \partial^2\bar\omega_1|^2 \, dx^7, \label{eq:50c} 
\end{align}
where in line \eqref{eq:50b} we have used the wave equation (\ref{eq:38}) and
in line \eqref{eq:50c} we have integrated by part and used that $\bar\omega_1$
decay  at infinity. In a similar way, we obtain
that energies with $n$-time derivatives control $n+1$ spatial derivatives, in particular
\begin{align}
  \label{eq:51}
 m_{\bar\omega}[\ddot{\bar\omega}_1] &\geq \int_{\mathbb{R}^7}
 |\partial^3\bar\omega_1|^2 \, dx^7, \\
m_{\bar\omega}[\dddot{\bar\omega}_1]&\geq \int_{\mathbb{R}^7}
 |\partial^4\bar\omega_1|^2 \, dx^7. \label{eq:51b}
\end{align}
Using the bound (\ref{eq:51}), (\ref{eq:51b}) and  Lemma \ref{l:estimaten} with $n=7$ and
$k=4$ the estimate (\ref{eq:1b}) follows.

\end{proof}

We finally remark that in the proof of the conservation of $m_2$ we have used
only the evolution equations for $\sigma_1$ and $\omega_1$, the time
derivative of the momentum constraint and the gauge equation for the shift.

\section{Extreme Kerr perturbations}
\label{s:kerr}
In this section we study the linearized equation obtained in section
\ref{linealizacion} for the case of extreme Kerr background. The main
difference with respect to the previous case of Minkowski is that the
background quantities $q_0, \sigma_0, \omega_0$ are not zero. However, we still
have that (see appendix \ref{sec:kerr-black-hole})
\begin{equation}
   \label{eq:46}
   \alpha_0=\rho.
 \end{equation}
 This is the main remarkably simplification of the extreme Kerr case compared
 with the non-extreme Kerr black hole.

 For the explicit form of $(q_0, \sigma_0, \omega_0)$ see the appendix
 \ref{sec:kerr-black-hole}. These functions depend on one parameter, the mass
 $m_0$ of the black hole. This mass is given by (\ref{eq:m0}).  The only
 properties of these functions that we will are the following. They satisfy the
 stationary equations
\begin{align}
\Ldt
\sigma_{0} &=\dfrac{\arrowvert\partial\omega_{0}\arrowvert^{2}}{\eta_{0}^{2}}, 
\label{eq:sta-sigma}\\
 \partial^{A}\left(\dfrac{\rho\partial_{A}\omega_{0}}{\eta_{0}^{2}}\right) &
 =0. \label{eq:sta-omega}  
\end{align}
They satisfy  the fall off conditions (\ref{eq:condicionesinfinito}),
(\ref{eq:condicionesend}). They satisfy the following inequality in $\Rdm$
(i.e. including both the origin and infinity)  (see \cite{Dain05d})
\begin{equation}
  \label{eq:falloffso}
  \frac{|\partial \omega_0|^{2}}{\eta^2_0} \leq \frac{C}{r^{2}}, \quad |\partial \sigma_0|^2 \leq \frac{C}{r^{2}},
\end{equation}
where $C$ is a constant that depends only on $m_0$. Finally, near the axis we have
\begin{equation}
  \label{eq:omega0axis}
\frac{\partial_\rho\omega_0}{\eta_0}=O(\rho).
\end{equation}

The complete set of linearized equations, in axial symmetry,  for the extreme Kerr black hole is the
following.

\textbf{Evolution equations:}

The evolution equations for $\sigma_1$ and $\omega_1$ are given by 
\begin{align}
 \label{eq:evolucion-sigma-kerr}
 -\frac{e^{2u_{0}}}{\rho^{2}}\dot{p}+^{(3)}\Delta\sigma_{1} &=\dfrac{2}{\eta^2_0}\left(\sigma_{1}|\partial\omega_{0}|^{2}-\partial_{A}
 \omega_{1}\partial^{A}\omega_{0} \right),\\
 \label{eq:evolucion-omega-kerr}
 -\frac{e^{2u_{0}}}{\rho^{2}}\dot{\po}+^{(3)}\Delta\omega_{1} &=4\dfrac{\partial_{\rho}\omega_{1}}{\rho}+
 2\partial_{A}\omega_{1}\partial^{A}\sigma_{0}+2\partial_{A}\omega_{0}\partial^{A}\sigma_{1},
\end{align}
with
\begin{align}
 \label{eq:pkerr}
 p &=\dot{\sigma}_{1}-2\dfrac{\beta_{1}^{\rho}}{\rho}-\beta_{1}^{A}\partial_{A}\sigma_{0},\\
 \label{eq:dkerr}
 \po &=\dot{\omega}_{1}-\beta_{1}^{A}\partial_{A}\omega_{0}.
\end{align}
The evolution equation for the metric and the second fundamental are obtained
replacing (\ref{eq:46}) in equations (\ref{eq:evolucionmetrica2}) and
(\ref{eq:evolucion-chi-2}). No relevant simplification occur in these equations
compared with the general expressions (\ref{eq:evolucionmetrica2}) and
(\ref{eq:evolucion-chi-2}), and hence we do not write them again in this
section. Also, we will not make use of these equations in the proof of theorem
\ref{theorem}.

\textbf{Constraint equations:}

The momentum constraint and Hamiltonian constraint  are given by
\begin{align}
 \label{eq:momento-kerr}
 \partial^{B}\chi_{1AB} &=-\dfrac{e^{2u_{0}}}{2\rho}\left(p\left(2\dfrac{\partial_{A}\rho}{\rho}
 +\partial_{A}\sigma_{0}\right)
+\dfrac{\partial_{A}\omega_{0}}{\eta_{0}^{2}}\po\right),\\
 \label{eq:hamiltonianokerr}
 ^{(3)}\Delta\sigma_{1}+\Delta q_{1} &=- \frac{\den_1}{4\rho },
\end{align}
where $\den_1$ is given by 
\begin{equation}
 \label{eq:epsilon2kerr}
 \den_1=\left( 2\partial_{A}\sigma_{0}\partial^{A}\sigma_{1} +\frac{2\partial_{A}\omega_{0}\partial^{A}\omega_{1}}{\eta_{0}^{2}}-\frac{2\sigma_{1}|\partial\omega_{0}|^2}{\eta_{0}^{2}}\right)\rho.
 \end{equation}

\textbf{Gauge equations:}

For the shift we have 
\begin{equation}
 \label{eq:shiftkerr}
 \left(\ck\beta_{1}\right)^{AB}=2e^{-2u_{0}}\rho \chi_{1}^{AB}.
\end{equation}

The energy density $\den_2$ defined previously in equation (\ref{eq:45}) is given by
\begin{multline}
 \label{eq:epsilonkerr}
 \den_{2}=\left(2\dfrac{e^{2u_{0}}}{\rho^{2}}p^{2}+ 2\dfrac{e^{2u_{0}}}{\rho^{2}\eta_{0}^{2}}\po^{2} 
 +4e^{-2u_{0}} \chi^{AB}_{1}\chi_{1AB} +\right. \\
\left. +2|\partial\sigma_{1}|^{2} +2\dfrac{\arrowvert\partial\omega_{1}\arrowvert^{2}}
 {\eta_{0}^{2}}
 +4\dfrac{\arrowvert\partial\omega_{0}\arrowvert^{2}}{\eta_{0}^{2}}\sigma_{1}^{2}
 -8\dfrac{\partial_{A}\omega_{0}\partial^{A}\omega_{1}\sigma_{1}}{\eta_{0}^{2}}\right)\rho
\end{multline}

Note that the energy density (\ref{eq:epsilonkerr}) is not
positive definite and hence it is by no means obvious that the energy $m_2$ is
positive. 

\begin{theorem}
  \label{theorem}
  Consider a smooth solution of the linearized equations presented above, such
  that it satisfies the fall off decay conditions at infinity
  \eqref{eq:condicionesinfinito}, the decay conditions at the extra asymptotic
  end at the origin \eqref{eq:condicionesend} and the regularity conditions
  \eqref{eq:condicionesparidad1}, \eqref{eq:condicionesparidad2} at the
  axis. Assume also that $\omega_1$ satisfies the following conditions. At the axis we have
  \begin{equation}
    \label{eq:omegaeje}
     \dot{\bar\omega}_1, \bar \omega_1=O_1(1),
  \end{equation}
and both at infinity and at the origin we impose
\begin{equation}
  \label{eq:omega1infinity}
  \dot{\bar\omega}_1, \bar \omega_1  =o_1(r^{-5/2}),
\end{equation}
where we have defined
\begin{equation}
  \label{eq:64}
  \bar \omega_1 =\frac{\omega_1}{\eta_0^{2}}.
\end{equation}

Then, we have:
\begin{itemize}
\item[(i)] The first order mass $m_1$ defined by (\ref{eq:m1}) with $\den_1$
  given by (\ref{eq:epsilon2kerr}) vanishes $m_1=0$.  The second order mass
  $m_2$ defined by (\ref{eq:m2}) with $\den_2$ given by (\ref{eq:epsilonkerr})
  is equal to the following expression, which is explicitly definite positive
  \begin{equation}
    \label{eq:60}
    m_2=\frac{1}{16}\int_{\Rdm} \bar\epsilon_2 \, \rho d\rho dz,
  \end{equation}
where  
\begin{multline}
    \label{eq:2} 
\bar\epsilon_2=\left(2\dfrac{e^{2u_{0}}}{\rho^{2}}p^{2}+ 2\dfrac{e^{2u_{0}}}{\rho^{2}\eta_{0}^{2}}\po^{2} 
 +4e^{-2u_{0}} \chi^{AB}_{1}\chi_{1AB} +\right. \\
+\left(\partial\sigma_{1}+\omega_{1}\eta^{-2}_{0}\partial\omega_{0}\right)^{2} +\left(\partial\left(\omega_{1}\eta_{0}^{-1}\right)-\eta_{0}^{-1}\sigma_{1}\partial\omega_{0}\right)^{2}+\\
\left. +\left(\eta_{0}^{-1}\sigma_{1}\partial\omega_{0}-\omega_{1}\eta_{0}^{-2}\partial\eta_{0}\right)^{2}\right)\rho. 
  \end{multline}

\item[(ii)] The mass $m_2$  is conserved along the evolution. 

\end{itemize}
\end{theorem}
Note that the boundary condition (\ref{eq:omegaeje}) at the axis (outside the
origin) is identical to the one used in Minkowski in section \ref{minkowki},
since $\eta_0$ behaves like $\rho^2$ at the axis. 

\begin{figure}
  \centering
\includegraphics{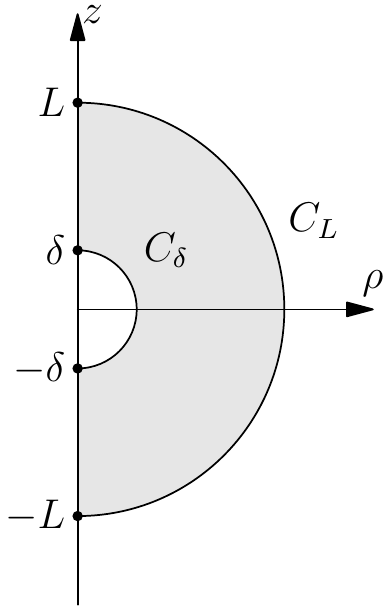} 
  \caption{Domain of integration in $\Rdm$ for the extreme Kerr black hole.}
  \label{fig:2}
\end{figure}

\begin{proof}
(i) We first prove that $m_1=0$. Take the density $\den_1$ given by~(\ref{eq:epsilon2kerr}), for the first term we have
\begin{align}
  \label{eq:3}
  2\rho \partial_A \sigma_0 \partial^A \sigma_1 &= 2 \partial^A\left(\rho \sigma_1 \partial_A \sigma_0  \right)- 2 \sigma_1 \partial^A \left(\rho \partial_A \sigma_0  \right),\\
&=  2 \partial^A\left(\rho \sigma_1 \partial_A \sigma_0  \right)-2\rho \sigma_1  \Ldt \sigma_0,\label{eq:3b}\\
&= 2 \partial^A\left(\rho \sigma_1 \partial_A \sigma_0  \right)-2 \rho \sigma_1 \frac{|\partial \omega_0|^2}{\eta_0^2},\label{eq:3c}
\end{align}
where in line (\ref{eq:3b}) we have used the definition of $\Ldt$ given by equation (\ref{eq:Ldt}) and in line (\ref{eq:3c}) we have used the stationary equation (\ref{eq:sta-sigma}). 

For the second term we have 
\begin{align}
\label{eq:om1}
\frac{2\rho\partial_{A}\omega_{0}\partial^{A}\omega_{1}}{\eta_{0}^{2}} &= 2\partial^A \left(\frac{\rho\omega_{1}\partial_{A}\omega_{0} }{\eta_{0}^{2}}  \right)  - 2 \omega_1 \partial^A \left(\frac{\rho \partial_A \omega_0}{\eta_0^2} \right),\\
&=  2\partial^A \left(\frac{\rho\omega_{1}\partial_{A}\omega_{0} }{\eta_{0}^{2}}  \right), \label{eq:om1b}
\end{align}
where in line (\ref{eq:om1b}) we have used the stationary equation (\ref{eq:sta-omega}). 
Summing up these terms we find
\begin{equation}
  \label{eq:27}
  \den_1= \partial_A t^A, 
\end{equation}
where
\begin{equation}
  \label{eq:47}
  t_A=2\rho \sigma_1 \partial_A \sigma_0+ 2\rho \omega_1  \frac{\partial_A \omega_0}{\eta_0^2}.
\end{equation}
We integrate equation (\ref{eq:27}) in the domain showed in figure \ref{fig:2}
for some finite $\delta$ and $L$ with $0<\delta< L$. At the axis the axis the
first term in (\ref{eq:47}) clearly vanished. The second term also vanishes by
the assumption (\ref{eq:omegaeje}) and the behavior (\ref{eq:omega0axis}) of
the background quantities. Hence, the integral of (\ref{eq:27}) contains only
the two boundary terms $C_\delta$ and $C_L$. Then, we take the limit $\delta\to
0$ and $L\to \infty$. Using the assumptions (\ref{eq:omega1infinity}) on
$\omega_1$, the assumptions (\ref{eq:condicionesinfinito}) and
(\ref{eq:condicionesend}) on $\sigma_1$ and the background fall off
(\ref{eq:falloffso}) we obtain that these two boundary integrals vanish. Hence,
it follows that $m_1=0$.

We prove now the positivity of $m_2$. The proof is identical to the proof of
positivity presented in section 3 of \cite{Dain05d}, which is based on the
Carter identity \cite{Carter71}. The last four terms in (\ref{eq:epsilonkerr})
are identical to the integrand of equation (24) in \cite{Dain05d} (in that
reference a different notation is used, namely $\sigma_{1}=\alpha$,
$\omega_{1}=y$, $\eta_{0}=X$ and $\omega_{0}=Y$). Then, Carter identity given
by equation (57) in \cite{Dain05d} in the notation of this article can be
written as
\begin{equation}
  \label{eq:62}
  \bar \epsilon_2 -\epsilon_2= \partial_A t^A,
\end{equation}
where
\begin{equation}
  \label{eq:63}
  t_A=2\rho \left(2\sigma_1 \partial_A \sigma_1 +\omega_1 \frac{\partial_A\omega_1}{\eta_0^2}
-2\sigma_1 \omega_1\frac{\partial_A \omega_0}{\eta^2_0}+\frac{\omega_1}{\eta_0} \partial_A \left( \frac{\omega_1}{\eta_0}\right)  \right),
\end{equation}
and $\bar \epsilon_2$ is given by (\ref{eq:2}).  Recall that the divergence
term in the right hand side of equation (\ref{eq:62}) has two contributions,
one is the right hand side of equation (57) in \cite{Dain05d} and the other
comes from the integration by parts in equation (63) in \cite{Dain05d}. Also
note that in \cite{Dain05d} Cartesian coordinates in $\mathbb{R}^3$ are used
for the integration, and here we use cylindrical coordinates, and hence the
factor $\rho$ appears in (\ref{eq:63}).  Integrating equation (\ref{eq:62}) and
using the fall off conditions at infinity and at the axis it follows that $m_2$
is given by (\ref{eq:60}) and hence it is positive.

(ii) To prove the conservation of $m_2$ we take a time derivative of the mass
density (\ref{eq:epsilonkerr}), we obtain
\begin{multline}
\label{eq:mkerr}
 \dot \varepsilon_2=4\dfrac{e^{2u_{0}}}{\rho}p\dot{p}+4\dfrac{e^{2u_0}}{\rho\eta_0^2} \po\dot\po + 8e^{-2u_{0}}\rho
\chi_{1}^{AB}\dot{\chi}_{1AB} +\\
+4\rho \partial_{A}
\sigma_{1}\partial^{A}\dot{\sigma}_{1} +
  4\rho
  \dfrac{\partial_{A}\omega_{1}\partial^{A}\dot{\omega}_{1}}{\eta_{0}^{2}}-8\rho
  \sigma_{1}
  \dfrac{\partial_{A}\omega_{0}}{\eta_{0}^{2}}\partial^{A}\dot{\omega}_{1}
+8\rho\dfrac{\arrowvert\partial\omega_{0}\arrowvert^{2}\sigma_{1}\dot{\sigma}_{1}}{\eta_{0}^{2}}
  -8\rho \dot{\sigma}_{1}\dfrac{\partial_{A}\omega_{0}}{\eta_{0}^{2}}\partial^{A}\omega_{1}.
\end{multline}
The strategy is very similar (but the calculations are more lengthy) than in
the Minkowski case discussed in section \ref{minkowki}: using the linearized
equations we will write the right hand side of (\ref{eq:mkerr}) as a total
divergence. We proceed analyzing term by term.  

For the first two terms we just use the definition of
$p$ and $d$ given in equations (\ref{eq:pkerr}) and (\ref{eq:dkerr}) respectively. We obtain
\begin{align}
  \label{eq:17}
  4\dfrac{e^{2u_{0}}}{\rho}p\dot{p} &=4\dfrac{e^{2u_{0}}}{\rho} \dot{p}
  \left(\dot \sigma_1 -\frac{2\beta^\rho_1}{\rho} -\beta_1^A \partial_A \sigma_0 \right),\\
\label{eq:18}
  4\dfrac{e^{2u_0}}{\rho\eta_0^2} \po\dot\po &=
  4\dfrac{e^{2u_0}}{\rho\eta_0^2} \dot\po \left( \dot\omega_1
    -\beta_1^A\partial_A \omega_0 \right).
\end{align}

For the third term we have
\begin{align}
  \label{eq:16}
  8e^{-2u_{0}}\rho \chi_{1}^{AB}\dot{\chi}_{1AB} &=
  8 \dot \chi_{1}^{AB} \partial_A \beta_{1B}, \\
 &= 8\partial_A(\beta_{1B}  \dot \chi_{1}^{AB})-8 \beta_{1B} \partial_A \dot
 \chi_{1}^{AB},\\
&= 8\partial_A(\beta_{1B}  \dot \chi_{1}^{AB})+4
\dfrac{e^{2u_{0}}}{\rho}\left(\dot{p} \left(2\dfrac{\beta_{1}^{\rho}}{\rho} 
 +\beta_{1}^{A}\partial_{A}\sigma_{0} \right)+\dot\po\dfrac{\beta_{1}^{A}\partial_{A}\omega_{0}}{\eta_{0}^{2}}\right),\label{eq:16c}
\end{align}
where in the  line \eqref{eq:16} we have used equation (\ref{eq:shiftkerr}) and the
fact that $\chi_1^{AB}$ is trace free and  in the line \eqref{eq:16c} we have used the
time derivative of equation (\ref{eq:momento-kerr}).

For the fourth term we have
\begin{align}
  \label{eq:19}
  4\rho \partial_{A} \sigma_{1}\partial^{A}\dot{\sigma}_{1} &= 4 \partial_{A} (
  \rho \dot \sigma_{1}\partial^{A} \sigma_{1})-4 \dot
  \sigma_1 \partial^A(\rho \partial_A \sigma_1), \\ 
& = 4 \partial_{A} (
  \rho \dot \sigma_{1}\partial^{A} \sigma_{1})-4 \dot
  \sigma_1 \rho \Ldt \sigma_1, \label{eq:19b}\\
&= 4 \partial_{A} (
  \rho \dot \sigma_{1}\partial^{A} \sigma_{1})-4\frac{\dot \sigma_1 \dot p
    e^{2u_0}}{\rho}  + 8\frac{\dot \sigma_1}{\eta^2_0} \rho \partial_A
  \omega_1 \partial^A \omega_0 - 8\frac{\dot\sigma_1 \sigma_1 \rho |\partial \omega_0|^2}{\eta^2_0},\label{eq:19c}
\end{align}
where in line \eqref{eq:19b} we used the definition of the operator $\Ldt$
given by equation (\ref{eq:Ldt}), in line \eqref{eq:19c} we used equation  \eqref{eq:evolucion-sigma-kerr}.

For the fifth term we obtain
\begin{align}
  \label{eq:20}
  4\rho \frac{\partial_A \dot \omega_1 \partial^A
    \omega_1}{\eta^2_0} &=4 \partial^A \left( \frac{\rho\dot \omega_1 \partial_A
      \omega_1}{\eta_0^2} \right)-4 \dot \omega_1 \partial^A \left( \frac{\rho\partial_A
      \omega_1}{\eta_0^2} \right),\\
& = 4 \partial^A \left( \frac{\rho\dot \omega_1 \partial_A
      \omega_1}{\eta_0^2} \right)-4\frac{\rho \dot \omega_1}{\eta_0^2}\left(
    \Ldt \omega_1-\frac{4\partial_\rho\omega_1}{\rho}-2\partial_A
    \omega_1\partial^A\sigma_0 \right), \label{eq:20b} \\
&= 4 \partial^A \left( \frac{\rho\dot \omega_1 \partial_A
      \omega_1}{\eta_0^2} \right) - 4\frac{e^{2u_0}}{\rho \eta^2_0}\dot \po
  \dot\omega_1-8\frac{\rho \dot \omega_1 \partial_A \omega_0 \partial^A
    \sigma_1}{\eta_0^2},\label{eq:20c} 
\end{align}
where in line \eqref{eq:20b} we have used the definition of the operator $\Ldt$ given in equation  (\ref{eq:Ldt}) 
and the definition of $\eta_0$ given in equation (\ref{eq:eta0sigma0}).  In the line \eqref{eq:20c} we have used the
evolution equation (\ref{eq:evolucion-omega-kerr}).

For the sixth term we obtain
\begin{align}
  \label{eq:21}
  -8\rho \sigma_{1}
  \dfrac{\partial_{A}\omega_{0}\partial^A\dot \omega_1}{\eta_{0}^{2}} &= -8 \partial^A
  \left( \rho \dot{\omega}_{1} \sigma_{1}
  \dfrac{\partial_{A}\omega_{0}}{\eta_{0}^{2}} \right) 
+8 \dot\omega_1 \partial^A \left( \rho \sigma_{1}
  \dfrac{\partial_{A}\omega_{0}}{\eta_{0}^{2}} \right),\\
&=-8 \partial^A
  \left( \dot{\omega}_{1} \sigma_{1}
  \dfrac{\partial_{A}\omega_{0}}{\eta_{0}^{2}} \right) +
8 \rho \dot \omega_{1}
\dfrac{\partial_{A}\omega_{0}\partial^{A} \sigma_{1}}{\eta_{0}^{2}}+8\dot\omega_1
\sigma_1 \partial^A \left(  \frac{\rho \partial_A \omega_0 }{\eta^2_0}\right),\\
&= -8 \partial^A
  \left( \dot{\omega}_{1} \sigma_{1}
  \dfrac{\partial_{A}\omega_{0}}{\eta_{0}^{2}} \right) +
8 \rho \dot \omega_{1}
\dfrac{\partial_{A}\omega_{0}\partial^{A} \sigma_{1}}{\eta_{0}^{2}}, \label{eq:21b}
\end{align}
where in line \eqref{eq:21b} we have used the stationary equation
(\ref{eq:sta-omega}). 

We sum the six terms obtained above plus the two last terms in
(\ref{eq:mkerr}),  only the divergence terms survive,  we obtain
\begin{equation}
  \label{eq:22}
  \dot \varepsilon_2= \partial_A t^A,
\end{equation}
where
\begin{equation}
  \label{eq:23}
  t^A= 4 
  \rho \dot \sigma_{1}\partial^{A} \sigma_{1}+ 8 \beta_{1B}  \dot \chi_{1}^{AB}+4  \frac{\rho\dot \omega_1 \partial^A
      \omega_1}{\eta_0^2}
-8  \dot{\omega}_{1} \sigma_{1}
  \dfrac{\partial^{A}\omega_{0}}{\eta_{0}^{2}}.
\end{equation}
Remarkably we get only one extra term compared with the Minkowski case (compare (\ref{eq:23}) with the sum of (\ref{eq:15}) and (\ref{eq:31})).

We integrate equation (\ref{eq:22}) in the domain showed in figure
\ref{fig:2}. The boundary term at the axis vanished by the hypothesis
(\ref{eq:condicionesparidad2}). Then, we take the limit $\delta\to 0$ and $L
\to \infty$, and the other two boundary integrals also vanished by the fall off
conditions (\ref{eq:condicionesinfinito}), (\ref{eq:condicionesend}) and (\ref{eq:omegaeje}), (\ref{eq:omega1infinity}).
  
\end{proof}

\section*{Acknowledgments}
 This work was supported in by grant PICT-2010-1387 of CONICET (Argentina) and
  grant Secyt-UNC (Argentina).

\appendix

\section{Kerr black hole in the maximal-isothermal gauge}
\label{sec:kerr-black-hole}
In this appendix we explicitly write the Kerr black hole metric in the maximal
-- isothermal gauge described in section \ref{sec:axially-symm-pert}. In
particular, we show that in this gauge the metric satisfies the conditions
(\ref{eq:beta0chi0}).

The Kerr metric, with parameters $(m,a)$, in Boyer-Lindquist coordinates
$(t,\tilde r,\theta,\phi)$ is given by
\begin{equation}
\label{kerr}
 g=- V dt^2+2 W dtd\phi
 +\frac{\Sigma}{\Dk} d \tilde r^2+\Sigma 
d \theta ^2+\eta d\phi^2,
\end{equation}
where
\begin{equation}
  \label{eq:app66}
\Dk =\tilde r^2+a^2-2m\tilde r,\quad
\Sigma=\tilde r^2+a^2 \cos^2 \theta,
\end{equation}
and
\begin{align}
  \label{eq:app67}
V &=\frac{\Dk-a^2\sin^2\theta}{\Sigma},\\  
W &=-\frac{2ma\tilde r\sin^2\theta}{\Sigma}, \label{eq:app68} \\
\eta & =\left(\frac{(\tilde r^2+a^2)^2 -\Dk a^2 \sin^2\theta}{\Sigma}
\right)\sin^2\theta.  \label{eq:app68b}
\end{align}
The angular momentum is given by
\begin{equation}
  \label{eq:57}
  J=ma.
\end{equation}

The metric (\ref{kerr}) is stationary and axially symmetric because it has the following two  Killing vectors
\begin{equation}
  \label{eq:app69}
 \xi^\mu=\left(\frac{\partial}{\partial t} \right)^\mu, \quad
\eta^\nu=\left(\frac{\partial}{\partial \phi} \right)^\nu,
\end{equation}
where $\xi^\mu$ is timelike near infinity (i.e. outside the ergosphere) and
$\eta^\mu$ is spacelike and it vanished at the axis.  The scalars
(\ref{eq:app67}), (\ref{eq:app68}) and (\ref{eq:app68b}) are written in terms
of the Killing vectors as follows
\begin{equation}
  \label{eq:app70}
 V= - \xi^\mu \xi^\nu g_{\mu\nu}, \quad \eta = \eta^\mu \eta^\nu g_{\mu\nu} , \quad W=\eta^\mu\xi^\nu g_{\mu\nu}. 
\end{equation}
In particular, $\eta$ is the square norm of the axial Killing vector
$\eta^\mu$. In these equations we are using 4-dimensional indices $\mu,\nu
\cdots$.

The twist potential $\omega$ of the axial Killing vector $\eta^\mu$  is given by
\begin{equation}
  \label{eq:18omega}
\omega = 2ma(\cos^3\theta-3\cos\theta)- \frac{2ma^3\cos\theta\sin^4\theta}{\Sigma}. 
\end{equation}

The 3-dimensional Lorenzian metric $h$ on the quotient manifold (see equation
(26) on \cite{Dain:2009wi}, we follow the notation of that article) is defined by
\begin{equation}
  \label{eq:app71}
  \eta g_{\mu\nu}= h_{\mu\nu}+\eta_\mu \eta_\nu. 
\end{equation}
Using the explicit form of the Kerr metric (\ref{kerr}) and the Killing vector
$\eta^\mu$ we obtain that $h$ is given by
\begin{equation}
  \label{eq:app29}
h= -(V \eta + W^2) dt^2+ \frac{\eta \Sigma}{\Dk} d \tilde
r^2+\eta \Sigma 
d \theta ^2. 
\end{equation}
For the Kerr metric, the following remarkably relation holds 
\begin{equation} 
\label{Kerrrelation}
V \eta +W^2=\Dk \sin^2 \theta.
\end{equation}
Using (\ref{Kerrrelation}) we further simplify the expression for the metric
$h$
\begin{equation}
  \label{eq:app31}
  h = -\Dk \sin^2 \theta dt^2+ \frac{\eta \Sigma}{\Dk} d \tilde
  r^2+\eta \Sigma 
  d \theta ^2.
\end{equation}
This metric is static. The foliation $t=constant$ has zero extrinsic curvature
and hence it is a maximal foliation. The shift of this foliation also vanished,
then the condition (\ref{eq:beta0chi0}) is satisfied.  However, the coordinates
$(\tilde r, \theta)$ are not isothermal because they do not satisfy the
condition (\ref{eq:2 metrica}).

To introduce isothermal coordinates we will assume that  $m\geq |a|$ (i.e. the
Kerr metric (\ref{kerr}) describe a black hole).   
Let  $r$ be defined as the positive root of the equation 
\begin{equation}
\tilde r =r +m+\frac{m^2-a^2}{4r},
\end{equation}
that is 
\begin{equation}
r = \frac{1}{2} \left(\tilde r -m + \sqrt{\Dk}  \right).
\end{equation}
We have
\begin{equation}
  \label{eq:appdr}
  d\tilde r = \frac{\sqrt{\Dk}}{r}dr .
\end{equation}
We define the  cylindrical coordinates $(\rho,z)$ in terms of the spherical
coordinates $(r,\theta)$ by the standard formula 
\begin{equation}
  \label{eq:67}
  \rho=r\sin\theta, \quad z=r\cos\theta. 
\end{equation}
Then the metric $h$ in the new
coordinate system $(t, \rho,z)$ is given by  
\begin{equation}
  \label{eq:app61}
  h= -\alpha^2 dt^2+e^{2u}(d\rho^2+ d z^2),
\end{equation}
where
\begin{equation}
  \label{eq:app62}
  \alpha=\sqrt{\Dk} \sin \theta=\rho\left(1-\frac{(m^2-a^2)}{4r^2} \right ),
\end{equation}
and  
\begin{equation}
  \label{eq:app33}
e^{2u}= \frac{\eta \Sigma}{r^2}. 
\end{equation}
The intrinsic metric of the $t=constant$  of the slices is
\begin{equation}
  \label{eq:69}
  q=e^{2u} \left(  d \rho^2+ d z^2\right). 
\end{equation}
That is, the coordinates system is isothermal.

The function $\sigma$ is defined in terms of the norm $\eta$ by 
\begin{equation}
  \label{eq:app35}
e^{\sigma}=\frac{\eta}{\rho^2},
\end{equation}

The function  $q$ is given  by
\begin{equation}
  \label{eq:app34}
e^{2q}= \frac{\sin^2\theta \Sigma}{\eta},
\end{equation}
We have the relation
\begin{equation}
  \label{eq:52}
 u=q+\sigma+\log \rho .
\end{equation}

Note that the lapse satisfies the  maximal gauge condition
\begin{equation}
  \label{eq:app63}
  \Ld \alpha=0.
\end{equation}
In the extreme case $m=|a|$ and hence we have
\begin{equation}
  \label{eq:66}
  \alpha=\rho. 
\end{equation}

\section{A Sobolev like estimate}
\label{sec:sobol-estim}

In this appendix we prove the following Sobolev type estimate.  

\begin{lemma}
\label{l:estimaten}
There exists a constant $C>0$ such that for all $u\in
C^\infty_0(\mathbb{R}^n)$, with $n\geq 3$, the following inequality holds
\begin{align}
  \label{eq:53}
 C\left(  \int_{\mathbb{R}^n} \left( |\partial^k u|^2  + |\partial^{k-1}
     u|^2\right) \, dx^n
 \right)^{1/2} \geq \sup_{x\in \mathbb{R}^n} |u(x)|,  
\end{align}
where $k>n/2$. 
\end{lemma}

\begin{proof}
  The estimate (\ref{eq:53}) will be a consequence of the following two
  classical estimates. The first one is the Gagliardo-Nirenberg-Sobolev
  inequality: assume that $1\leq p <n$, then exists a constant $C$, depending
  only on $p$ and $n$, such that
\begin{align}
  \label{eq:54}
  ||u||_{L^{q}(\mathbb{R}^n)}\leq C ||\partial u||_{L^{p}(\mathbb{R}^n)},
\end{align}
for all $u\in C^\infty_0(\mathbb{R}^n)$, where  
\begin{equation}
  \label{eq:55}
  q=\frac{pn}{n-p}.
\end{equation}

The second estimate is the Morrey's inequality: assume $n<p\leq \infty$, then
there exists a constant depending only on $p$ and $n$, such that
\begin{equation}
  \label{eq:58}
  \sup_{x\in \mathbb{R}^n} |u(x)|\leq C ||u||_{W^{1,p}(\mathbb{R}^n)}.
\end{equation}
See \cite{Evans98} for an elementary and clear presentation of these
inequalities and the functional spaces $L^{p}(\mathbb{R}^n)$,
$W^{1,p}(\mathbb{R}^n)$ involved in them.

We first observe that the estimate (\ref{eq:54}) can be iterated as follows: 
\begin{equation}
  \label{eq:68}
  ||u||_{L^{p_k}(\mathbb{R}^n)}\leq C ||\partial^k u||_{L^{p}(\mathbb{R}^n)},
\end{equation}
where $1 \leq k\leq n/p$, $1<p$ and $p_k$ is given by
\begin{equation}
  \label{eq:70}
  p_k=\frac{pn}{n-pk}.
\end{equation}
To prove (\ref{eq:68}) we use induction in $k$. For $k=1$ the inequality
(\ref{eq:68}) reduces to (\ref{eq:54}). Assume that (\ref{eq:68}) is valid for
$k$. If $\partial^{k+1}u \in L^{p}(\mathbb{R}^n)$, then by (\ref{eq:54}) we
obtain that $\partial^k u \in L^{q}(\mathbb{R}^n)$ with $q$ given by
\begin{equation}
  \label{eq:71}
  q=\frac{pn}{n-p}.
\end{equation}
By the inductive hypothesis we obtain that $u\in L^{q_k}(\mathbb{R}^n)$
with
\begin{equation}
  \label{eq:72}
  q_k=\frac{qn}{n-qk}.
\end{equation}
We substitute  (\ref{eq:71}) in (\pageref{eq:72}) to obtain
\begin{equation}
  \label{eq:73}
  q_k=\frac{pn}{n-(k+1)p}.
\end{equation}
And then the desired result is proved.

To prove (\ref{eq:53}), note that the left hand side of (\ref{eq:53}) implies
that $\partial^{k-1} w , \, \partial^{k-1} u \in L^2(\mathbb{R}^n)$ where
$w=\partial u$. Then, we apply the inequality (\ref{eq:68}) for both $w$ and
$u$, to obtain that $w,u\in L^p (\mathbb{R}^n)$, with $p$ given by
 \begin{equation}
   \label{eq:56}
   p=\frac{2n}{n-2k+2}.
 \end{equation}
By hypothesis $k>n/2$, then  we obtain that $p>n$. Hence, we have proved that
$u\in W^{1,p}(\mathbb{R}^n)$ with $p>n$. We use the Morrey inequality (\ref{eq:58}) and
the desired result follows.

\end{proof}


\end{document}